%% file: NQIV-spec.tex
\documentclass[11pt,a4paper,fleqn]{scrartcl}
\RequirePackage{amsthm,amsmath,natbib}
\RequirePackage{amssymb,amstext}
\RequirePackage{hypernat}
\usepackage[pdfpagemode=UseOutlines ,plainpages=false
,hypertexnames=false ,pdfpagelabels ,hyperindex=true,colorlinks=true]{hyperref}
	\makeatletter%
	\Hy@breaklinkstrue%
	\makeatother%
\usepackage{color}
\usepackage{stmaryrd}
\definecolor{darkred}{rgb}{0.6,0.0,0.1}
\definecolor{darkgreen}{rgb}{0,0.5,0}
\definecolor{darkblue}{rgb}{0,0,0.5}
\hypersetup{colorlinks ,linkcolor=darkblue ,filecolor=darkgreen
,urlcolor=darkblue ,citecolor=black}

\usepackage{url}
\usepackage{verbatim}
\usepackage{ bm}
\usepackage{graphicx}
\usepackage{booktabs,makecell}

\usepackage{pgfplots}
\pgfplotsset{compat=1.10}
\usepgfplotslibrary{fillbetween}
\usetikzlibrary{intersections}

\renewcommand{\cite}{\citet}
\usepackage{abrev-package}

\theoremstyle{mysc}\newtheorem{prop}{Proposition}[section]
\theoremstyle{mysc}
\theoremstyle{mysc}\newtheorem{coro}[prop]{Corollary}
\theoremstyle{mysc}\newtheorem{theo}[prop]{Theorem}
\theoremstyle{mysc}
\theoremstyle{mysc}\newtheorem{lem}[prop]{Lemma}
\theoremstyle{myex}\newtheorem{rem}{Remark}[section]
\theoremstyle{myex}\newtheorem{example}{Example}[section]
\theoremstyle{myex}

\newtheorem{assA}{Assumption}

\theoremstyle{mysc}
\theoremstyle{mysc}
\theoremstyle{mysc}

\numberwithin{equation}{section}

\makeatletter%
\def\@fnsymbol#1{\ensuremath{\ifcase#1\or * \or \star \or 1 \or 2\or 3\or  , \or
g\or h\or i\else\@ctrerr\fi}}%
\makeatother%

 \author{\textsc{Christoph Breunig}\thanks{Department of Economics, Emory University, Rich Memorial Building, Atlanta, GA 30322, USA, e-mail:
\url{	christoph.breunig@emory.edu}}\\
{\small \textit{Emory University}}
}

\title{{Specification Testing in Nonparametric Instrumental Quantile Regression}\thanks{Parts of this paper derive from my doctoral dissertation, completed under the
guidance of Enno Mammen.  I would like to thank Liangjun Su and four anonymous referees for excellent comments and suggestions that greatly improved the paper. I also thank  seminar participants at Boston College, Bristol, Mannheim,  Toulouse School of Economics, University College London, WIAS Berlin, and Yale. I am  also  grateful for the support and hospitality of the Cowles Foundation.}
}

\begin{document}
\maketitle
\begin{abstract}
There  are many environments in econometrics which require  nonseparable modeling of a structural disturbance. In a nonseparable model with endogenous regressors, key conditions are validity of instrumental variables and monotonicity of the model in a scalar unobservable variable. Under these conditions the nonseparable model is equivalent to an
instrumental quantile regression model.
A failure of the key conditions, however, makes instrumental quantile regression potentially inconsistent.
This paper develops a methodology for testing the hypothesis whether the instrumental quantile regression model is correctly specified.
Our test statistic is asymptotically normally distributed under correct specification and consistent against any alternative model. In addition, test statistics to justify the model simplification are established.
 Finite sample properties are examined in a Monte Carlo study and an empirical illustration is provided.
\end{abstract}
\begin{tabbing}
\noindent \emph{Keywords:} \=Nonparametric quantile regression, instrumental variable,\\ 
\> specification test, local alternative, nonlinear inverse problem.\\[.2ex]
\> 
\end{tabbing}

\input{introduction}
\input{main-results}
 \bibliography{BiB}

\appendix
\section{Appendix}\label{app:proofs} 
 \subsection{Proofs of Section \ref{sec:main}.}
 \input{proofs-asymt}
 \input{proofs-tech}

\end{document}

%% file: introduction.tex
\section{Introduction}
Regression models that involve instrumental variables are widely used in economics to overcome endogeneity problems.
In these models, assuming the structural disturbances to be  additively separable implies that marginal effects do not depend on unobserved characteristics, which may be difficult to justify.
 This is why their nonseparable extension has received a lot of attention recently.
Under certain key conditions,
the nonseparable model is equivalent to an instrumental quantile regression model.  
These conditions are validity of instruments and monotonicity of the model in a scalar unobservable.
 If one of these conditions is violated, however, the quantile regression representation is misspecified. 

In this paper, we propose a specification test of  the instrumental quantile regression model
	\begin{equation}
	\label{model:NP}
	Y = \sol(Z,q)+U(q)\quad\mbox{ where }\quad \PP(U(q)\leq 0|W)= q
	\end{equation}
	for each $0<q<1$, where $Y$ is a scalar dependent variable, $Z$ a vector of potentially endogenous regressors, $W$ a vector of instruments, and $U(q)$ an unobservable disturbance.\footnote{Since conditional expectations are defined only up to equality almost surely, all (in)equalities with
conditional expectations and/or random variables are understood as (in)equalities almost surely.}
This quantile regression model is equivalent to a  nonseparable model (cf. \cite{HL07econometrica}) given by
\begin{equation}
\label{model:NP:nonsep}
 Y=\sol(Z,V)
\end{equation}
where
\begin{itemize}
 \item[\textit{(a.1)}] the instrumental variable $W$ is independent of $V$,
 \item[\textit{(a.2)}] the function $\sol$ is strictly monotonic increasing in the scalar disturbance $V$, and
 \item[\textit{(a.3)}] $V \sim \cU (0,1)$.
\end{itemize}
Condition \textit{(a.3)} can be assumed without loss of generality if $V$ is continuously distributed with positive density on its support which we assume to hold throughout the paper. The quantile regression model \eqref{model:NP} for all $0<q<1$ is thus misspecified if in its nonseparable version \eqref{model:NP:nonsep} the instrument is not valid, that is, $W$ is not independent of $V$, or the function $\sol$ is not monotonic in $V$.

 Specification testing in instrumental variable models is a subject of considerable literature.
In the context of nonparametric instrumental mean regression $Y=g(Z)+U$ with $\Ex[U|W]=0$, tests for correct specification have been proposed by \cite{GS07}, \cite{Horowitz2009}, and \cite{Breunig2012}. These tests are, however, not robust against potential nonseparability of the structural disturbance. On the other hand, by considering the nonseparable model \eqref{model:NP:nonsep} with conditions \textit{(a.1)}--\textit{(a.3)} a failure of the exclusion restriction of the instruments might only be one source of misspecification. Indeed, as argued by \cite{hoderlein2007}, in certain applications, such as consumer demand, the monotonicity restriction \textit{(a.2)} might be highly unrealistic. As such, providing a specification test of model \eqref{model:NP:nonsep} together with conditions \textit{(a.1)}--\textit{(a.3)} seems paramount but, as far as we know, has not yet been addressed in the literature.

Research on identification and estimation in nonparametric instrumental quantile regression has been active in the last decade. 
\cite{2003Chesher} establishes nonparametric identification of derivatives of the unknown functions in a triangular array structure. \cite{Chern2005} and \cite{Chern2007} give identification conditions and develop a nonparametric minimum distance estimator. Sufficient conditions for local identification are given by \cite{chen2012}.
\cite{HL07econometrica} propose an estimator based on Tikhonov regularization, \cite{chen08} study penalized sieve minimum distance estimation, and \cite{Dunker2011} consider an iteratively regularized Gau\ss-Newton method.
Further, \cite{Gag2012} obtain asymptotic distribution results of a Tikhonov regularized estimator. There is also a large literature on testing quantile regression models with exogenous covariates. In this context particularly relevant is  quantile regression testing using an infinite number of quantiles for parametric functions, see \cite{escanciano2010} and, in the nonparametric context, \cite{escanciano2014}.

In instrumental quantile regression \eqref{model:NP} for a fixed quantile $0<q<1$, \cite{HorLee2009} established a test of parametric specification of $\sol$. \cite{chen2015} consider functionals of semi/nonparametric conditional moment restrictions with possibly nonsmooth generalized residuals. A test of monotonicity in unobservables of $\sol$ has been proposed by \cite{hoderlein2011} but requires conditional exogeneity of $Z$ and hence, is not related to instrumental variables methodology. Recently and independently of this paper, \cite{feve2013} developed a test of whether $Z$ is independent of the nonseparable disturbance $V$ in the model \eqref{model:NP:nonsep}.

Our test statistic is based on the $L^2$--norm of the empirical conditional quantile restriction and involves sieve methodology. The sieve approach makes the statistic easy to implement and further, is convenient to impose additional constraints on the structural function $\sol$. As an example, we  discuss a test of additivity of $\sol$ with respect to the vector of regressors $Z$. In addition, we establish a test statistic for testing exogeneity which is robust against nonseparability.  More precisely, we establish a test of exogeneity of the regressors $Z$ at some quantile $0<q<1$, that is, whether $\PP(Y\leq \sol(Z,q)|Z)=q$.
This extends the results on nonparametric tests of exogeneity in mean regression suggested by \cite{Blundell2007} and \cite{Breunig2012} to the quantile regression case.

 It should also be noted that the test proposed in this paper is a joint test of monotonicity and instrument validity. This is the nature of many nonparametric tests, see, for instance, \cite{chiappori2015} or \cite{lewbel2015}. On the other hand, we show in this paper how the sign of $\PP(Y\leq \sol(Z,q)|W)-q$ can be exploited to make inferences on the validity of the instrumental variables. As such, in many cases it is possible to detect the cause of a rejection of our test.

 We establish the asymptotic distribution of our test statistic under the null hypothesis and its consistency against fixed alternatives. We study the power of our test against a sequence of local alternatives. By Monte Carlo simulations we demonstrate the power properties of our test in finite samples. As an empirical illustration, we study a nonseparable model of the effects of class size on test scores of 4th grade students in Israel. We reject the hypothesis of exogeneity of class size  but fail to reject the instrumental variable model.

The remainder of this work is organized as follows. In Section \ref{sec:main}, we propose a test statistic and obtain its asymptotic distribution. We further establish consistency of our test. The power of the test is judged by considering a sequence of local alternatives.
Section \ref{sec:extension} gives several extensions of the previous results.
In Section \ref{sec:MC} and \ref{sec:emp} we study the finite sample properties of our test and give an empirical illustration. All proofs can be found in the appendix.

%% file: main-results.tex
 \section{The test statistic and its asymptotic properties}\label{sec:main}
This section begins with the definition of the test statistic and states assumptions required to obtain its asymptotic distribution under the null hypothesis.
Moreover, we study power and consistency properties of our test.

\subsection{Definition of the test statistic}
The quantile regression model \eqref{model:NP} leads to a nonlinear operator equation, as we see in the following. 
Let $\varPhi$ be a Banach space endowed with the norm $\|\phi\|_{Z,p}:=(\Ex|\phi(Z)|^p)^{1/p}$ for some integer $p>0$ and if $p=\infty$ then $\|\phi\|_{Z,\infty}:=\sup_{z}|\phi(z)|$. For simplicity let $\|\phi\|_Z:=\|\phi\|_{Z,2}$. Further, let us introduce the Hilbert space $L_W^2:=\{\psi:\,\|\psi\|_W^2:=\Ex|\psi(W)|^2<\infty\}$. We define a nonlinear operator $\cT:\varPhi\to L_W^2$ with
	\begin{equation}
	\label{model:preEE}
	\cT\phi:=\Ex[\1{\{Y\leq \phi(Z)\}}|W]
	\end{equation}
for any $\phi\in\varPhi$ where $\1$ denotes the indicator function. Thereby, model \eqref{model:NP} can be rewritten as the operator equation $\cT\solq=q$ with $\sol_q(\cdot):=\sol(\cdot,q)$ for all $0<q<1$. 

 In many economic applications, for instance when estimating a demand function or Engel curves, the structural function of interest may be assumed to be smooth. This \textit{a priori} knowledge is captured by a set $\cB\subset \varPhi$ which we introduce below. The set $\cB$ may also contain constraints on the function $\solq$ such as monotonicity, concavity/convexity or additivity (see also Section \ref{sec:add}) and can also ensure uniqueness of $\solq$ (see Example \ref{exmp:ident} below). Let us introduce the set $\cB^{(0,1)}=\set{\phi:\, \phi(\cdot,q)\in\cB \text{ for all }q\in(0,1)}$. We consider the null hypothesis
	\begin{equation}
	\label{model:EE}
	H_0:\,\text{there exists a function }\sol\in\cB^{(0,1)} \text{ such that } \cT\solq=q \text{ for all }q\in(0,1).
	\end{equation}
The alternative is that there exists no function $\sol\in\cB^{(0,1)}$ solving $\cT\solq=q$ for all $q\in(0,1)$.

 We construct in the following a test statistic based on the $L^2$--distance.
 Throughout the paper, we assume that an independent and identically distributed $n$-sample of $(Y,Z,W)$ is available. 
 Let $\{f_j\}_{j\geq 1}$ be a sequence of approximating functions in $L_W^2$. Then, for any integer $k\geq 1$ we denote $f_\uk(\cdot)=(f_1(\cdot),\dots,f_k(\cdot))^t$ and $\bold W_k=\big(f_\uk(W_1),\dots,f_\uk(W_n)\big)^t$ which is a $n\times k$ matrix.
 A series least square estimator of $\Ex[\1\set{Y\leq \phi(Z)}-q|W=\cdot]$ then writes
 \begin{equation*}
 f_\uln(\cdot)^t(\bold W_\l^t\bold W_\l)^-\sum_{i=1}^n(\1\set{Y_i\leq \phi(Z_i)}-q)f_\uln(W_i)
 \end{equation*}
 where $(\cdot)^-$ denotes a generalized inverse. 
 Further, we define the sieve least square estimator of $\solq$ by
 \begin{equation}\label{def:est}
  \hsolq\in\argmin_{\phi\in\cB_\k}\Big(\sum_{i=1}^n(\1\set{Y_i\leq \phi(Z_i)}-q)f_\uln(W_i)\Big)^t(\bold W_\l^t\bold W_\l)^-\sum_{i=1}^n(\1\set{Y_i\leq \phi(Z_i)}-q)f_\uln(W_i)
 \end{equation}
where $\cB_\k$ is a $\k$--dimensional sieve space that becomes dense in $\cB$ as the sample size $n$ tends to infinity. If $\cB$ contains additional constraints then these are imposed in $\cB_\k$ on the finite dimensional functions. 
Here, $\k$ and $\l$ grow with sample size $n$. Clearly, $\k\leq \l$ for each $n$ is required and in our simulations we choose $\l = C\k$ for some constant $C>1$ (see also \cite{chen2015optimal} in the case of nonparametric instrumental mean regression). The estimator $\hsol_{qn}$ is a simplified version of the penalized sieve minimum distance estimator suggested by \cite{chen08}.

The test statistic is then given by
\begin{equation}\label{sn}
S_n=\int_0^1\Big(\sum_{i=1}^n(\1\set{Y_i\leq \hsolq(Z_i)}-q)f_\umn(W_i)\Big)^t(\bold W_\m^t\bold W_\m)^-\sum_{i=1}^n(\1\set{Y_i\leq \hsolq(Z_i)}-q)f_\umn(W_i)dq
\end{equation}
where $\m$ grows with sample size $n$. As the test is one sided, we reject the null hypothesis at level $\alpha$ when the standardized version of $S_n$, namely $3\sqrt{5/\m}\big(S_n-\m/6\big)$, is larger than the $(1-\alpha)$--quantile of $\cN(0, 1)$.
 The asymptotic distribution of $S_n$ is derived below under mild restrictions on the dimension parameters $\k$, $\l$, and $\m$. We require that the number of unconditional moment restrictions determined by $\m$ is asymptotically larger than the dimension of the sieve space $\cB_\k$. This corresponds to the test of overidentifying restrictions in parametric models. In contrast to the parametric setting, however, also the number of unconditional moment restrictions used to construct the estimator (determined by $\l$) must be asymptotically smaller than the number of moment restrictions used for the test statistic. This ensures that the estimation error in the test statistic becomes asymptotically negligible as we see below. 

 Our test statistic builds on the nonparametric specification test in instrumental mean regression suggested by \cite{Breunig2012}. Testing in instrumental quantile regression, on the other hand, requires a different methodology. First, the test statistic is a discontinuous function of the unknown structural effect $\solq$. Second, instrumental quantile regression leads to a nonlinear inverse problem and hence deriving asymptotic results is more challenging. Third, to verify the conditional moment restrictions for all quantiles we need to integrate over them.
 In the appendix, we show that the mapping $q\mapsto \solq$ is continuous under mild assumptions. This justifies the use of our $L^2$--type rather than a sup norm statistic.

\subsection{Assumptions and notation}
In order to obtain our asymptotic result we state the following assumptions.
Our first assumption gathers conditions which we require for the basis functions $\{f_j\}_{j\geq1}$. In the following, the supports $\cZ$ of $Z$ and $\cW$  of $W$  are assumed to be bounded (see also Assumption \ref{ass:A0}).  The probability density function (p.d.f.) of $W$, denoted by $p_W$, is assumed to be uniformly bounded from above and away from zero. 

\begin{assA}\label{ass:A1:q} 
(i) There exists a constant $C>0$ and a sequence of positive integers $(m_n)_{n\geq 1}$ satisfying  $\sup_{w\in\cW}\|f_\umn(w)\|^2\leqslant C\m$.
 (ii) The smallest eigenvalue of the matrix $\Ex[f_\um(W)f_\um(W)^t]$ is bounded away from zero uniformly in $m$. 
\end{assA}
 Assumption \ref{ass:A1:q} $(i)$ holds for sufficiently large $C$   if  $\{f_j\}_{j\geq1}$ are  trigonometric basis functions,  B-splines, or wavelets. Assumption \ref{ass:A1:q} $(ii)$ is satisfied if the marginal density of $W$ is uniformly bounded away from zero on $\mathcal W$ and $f_{\umn}$ forms a vector of orthonormal basis functions. 
For any $\phi\in\cB^{(0,1)}$ we write $\phi_q(\cdot):=\phi(\cdot,q)$ for all $0<q<1$.
We denote the Fr\'echet derivative of $\cT$ at $\solq$ by
\begin{equation*}
 T_q\phi:=\Ex\big[ p_{Y|Z,W}\big(\sol(Z,q),Z,W\big)\phi(Z)\big|W\big]
\end{equation*}
where $p_{Y|Z,W}$ denotes the density of $Y$ conditional on $(Z,W)$.  We introduce the notation $\interleave\phi\interleave_{Z,p}=\big(\int_0^1\|\phi(\cdot,q)\|_{Z,p}^pdq\big)^{1/p}$ and $\interleave\psi\interleave_W=\big(\int_0^1\|\psi(\cdot,q)\|_W^2dq\big)^{1/2}$ for functions $\phi(\cdot,q)\in\varPhi$ and $\psi(\cdot,q)\in L_W^2$ for all $q\in(0,1)$.
\begin{assA}\label{ass:A3:q}(i) If $\interleave\cT\phi-\cT\sol\interleave_W^2=0$ for some function $\phi\in\cB^{(0,1)}$ then it holds $\interleave\phi-\sol\interleave_{Z,p}^2=0$.  (ii) There exists some constant $0<\eta<1$ such that for all $0<q<1$ and all functions $\phi\in\set{\phi\in\cB:\,\|\phi-\sol_q\|_{Z,p}\leq \varepsilon}$ for some $\varepsilon>0$ it holds
 \begin{equation}\label{tang:cone:cond}
  \|\cT\phi-\cT\solq-\op_q(\phi-\solq)\|_W\leq \eta\|T_q(\phi-\solq)\|_W.
 \end{equation}
\end{assA}
Assumption \ref{ass:A3:q} $(i)$ ensures identification of $\solq$ for almost all $0<q<1$ on the set $\cB$ which we introduce below.
Assumption \ref{ass:A3:q} $(ii)$ specifies an upper bound on the Taylor remainder of $\cT$ in a small neighborhood around $\solq$. It is also known as the tangential cone condition and frequently used in the analysis of nonlinear operator equations (cf. \cite{Hanke1995} or \cite{Dunker2011} in case of instrumental variable estimation). 
We provide sufficient conditions for the tangential cone condition in Example \ref{exmp:ident} below and refer to   \cite{chen2012} for further discussions. 
\begin{assA}\label{ass:A5:q}
 There exists a sequence $(r_n)_{n\geq 1}$ with $r_n=o(1)$ such that for constants $C>0$ and $\kappa\in(0,1]$ it holds
\begin{equation}\label{eq:kappa}
 \max_{1\leq j\leq \m}\Ex\Big[\int_0^1\sup_{\phi\in\cB_{n}}\big|\1\{Y\leq\phi(Z,q)\}-\1\{Y\leq \sol(Z,q)\}\big|^2dq\,f_j^2(W)\Big]
\leq Cr_n^{2\kappa}
\end{equation}
where $\cB_{n}:=\{\phi\in\cB^{(0,1)}:\,\interleave\phi-\sol\interleave_{Z,p}^2\leq r_n^2\}$.
\end{assA}
Assumption \ref{ass:A5:q}  states that the function $\solq\mapsto(\1\{Y\leq\sol(Z,q)\}-q)f_j(W)$, $1\leq j\leq \m$,  is locally uniformly $L_W^2$ continuous for almost all $0<q<1$. This condition has also been exploited by \cite{CLVK03econometrics} (Theorem 3), \cite{Chen07} (Lemma 4.2 (i)) or \cite{chen08} (Remark c.1).
Example \ref{exmp:loc:cont} below gives primitive conditions under which Assumption \ref{ass:A5:q} holds true.

Let $\cZ\subset \mathbb R^{d_z}$ and for any vector of nonnegative integers $k=(k_1,\dots,k_{d_z})$  define $|k|=\sum_{j=1}^{d_z}k_j$ and $D^k=\delta^{|k|}/(\delta z_1^{k_1}\dots\delta z_{d_z}^{k_{d_z}})$. For some integer $p>0$ we define the norms
\begin{equation*}
 \|\phi\|_{\alpha,p}=\Big(\sum_{|k|\leq \alpha+\alpha_0}\int_\cZ \big|D^k\phi(z)\big|^p dz\Big)^{1/p},\quad \|\phi\|_{\alpha,\infty}=\max_{|k|\leq \alpha}\sup_{z\in\cZ}\big|D^k\phi(z)\big|
\end{equation*}
where $\alpha$ and $\alpha_0$ are positive integers. We denote the Sobolev spaces associated with the norm $ \|\cdot\|_{\alpha,p}$ by 
\begin{equation}
 W^{\alpha,p}:=\set{\phi:\cZ\to\mathbb R:\,\|\phi\|_{\alpha,p}<\infty}.
\end{equation}
For some constant $\rho>0$, define $\cB$ as the Sobolev ellipsoid of radius $\rho$ given by
\begin{equation}
 \cB:=\cB(\alpha):=\set{\phi\in W^{\alpha,p}:\,\|\phi\|_{\alpha,p}\leq \rho}.
\end{equation}
On the other hand, our sieve space $\cB_\k$ used to approximate $\cB$ is compact under $\|\cdot\|_Z$ and thus, penalization is not necessary for consistent estimation (see also \cite{chen08}). Also additional constraints such as monotonicity can be imposed by $\cB=\set{\phi\in W^{\alpha,p}:\,\|\phi\|_{\alpha,p}\leq \rho, \,\inf_{z\in\cZ}\phi'(z)>0}$ for  scalar $z$. Such a monotonicty constraint does not necessarily lead to faster rates of convergence, in contrast to an additivity restriction on $\solq$.  Consequently, we do not treat shape restrictions like monotonicty explicitly but only discuss a test of additivity in Section \ref{sec:add}. In this context, we also refer to \cite{chetverikov2015} for using shape restriction for sieve estimation in instrumental mean regression.
The following assumption gathers regularity conditions imposed on the structural functions $\sol$ and the supports $\cZ$ and $\cW$. 
\begin{assA}\label{ass:A0}
(i) Let $\alpha_0>{d_z}/p$ and $\alpha>{d_z}/\kappa$.
(ii) $\cZ$ is  bounded, convex and satisfies a uniform cone property. (iii) $\cW$ is bounded. (iv) The marginal density of $W,$ denoted by $p_W$, is bounded from above and uniformly bounded away from zero on $\cW$.
 (v) $p_{Y|Z,W}(\cdot,Z,W)$ is  bounded from above. 
\end{assA}
Assumption \ref{ass:A0} $(i)$ requires $\alpha$ to be large if \eqref{eq:kappa} holds only for small $\kappa>0$ or the dimension ${d_z}$ is large.
Assumption \ref{ass:A0} $(ii)$ imposes a weak regularity condition on the shape of $\cZ$. For the uniform cone property see, for instance, Paragraph 4.4 in \cite{2003Adams}. This property was also used by \cite{Santos12}. Assumption \ref{ass:A0} $(v)$ ensures that $\|T_q\phi\|_W\leq C\|\phi\|_Z$ for all $\phi\in L_Z^2$ and some constant $C>0$.
\begin{example}[Primitive Conditions for Assumption \ref{ass:A3:q}]\label{exmp:ident}
 Let $\varPhi$ coincide with the Hilbert space $L_Z^2:=\set{\phi:\,\|\phi\|_Z<\infty}$. If for any $0<q<1$ the operator $T_q$ is compact then there exists an orthonormal basis in $L_Z^2$ denoted by $\{e_j\}_{j\geq 1}$ satisfying $\|T_q\phi\|_W^2=\sum_{j=1}^\infty s_{qj}^2\Ex[\phi(Z)e_j(Z)]^2$ where $(s_{qj})_{j\geq 1}$ are the singular values of $T_q$. 
 If 
 \begin{equation*}
 \cB\subset\cB_{source,q}:=\set{\phi\in L_Z^2:\,\sum_{j=1}^\infty s_{qj}^{-2}\Ex[(\phi(Z)-\sol(Z,q))e_j(Z)]^2<c_0}  
 \end{equation*}
for some constant $c_0>0$ then, under mild assumptions on the joint distribution of $(Y,Z,W)$, the function $\solq$ is identified on $\cB$ (cf. Theorem 6 of \cite{chen2012}). A similar restriction was also imposed by \cite{HL07econometrica}.  If $\cB\subset \bigcap_{q\in(0,1)}\cB_{source,q}$ then Assumption \ref{ass:A3:q} $(i)$ holds true. Under further assumptions, imposing bounds on the generalized Fourier coefficients is equivalent to imposing smoothness restrictions. To illustrate this relation let $Z$ be a scalar uniformly distributed random variable and assume $s_{qj}=j^{-\zeta}$, $j\geq 1$, for some constant $\zeta>0$. In this case, if $\set{e_j}_{j\geq 1}$ are the usual trigonometric basis functions then $\cB_{source,q}$ coincides with the Sobolev space of $\zeta$--times differentiable functions with periodic boundary conditions, while if $s_{qj}^2=\exp(-j^{2\zeta})$, $j\geq 1$ and $\zeta>1$,  $\cB_{source,q}$ contains only analytic functions (see also \cite{Kr89}). 
 In this sense, $\cB_{source,q}$ links the smoothness of $\phi-\sol_q$ to the degree of ill-posedness determined by the degree of decay of $(s_{qj})_{j\geq 1}$, which is also known as a so-called \textit{source condition} (cf. \cite{ChenReiss2008} or \cite{Dunker2011} for a further discussion). 

Under the singular value decomposition of $T_q$ it is also possible to provide primitive conditions for the tangential cone condition \eqref{tang:cone:cond}. Assume that the conditional p.d.f. of $Y$ given $(Z,W)$, denoted by $p_{Y|Z,W}$, is continuously differentiable with $|\partial p_{Y|Z,W}(\cdot,Z,W)/\partial y|\leq c_1$ and the conditional p.d.f. of $Z$ given $W$ satisfies $p_{Z|W}(\cdot,W)\leq c_2p_{Z}(\cdot)$, for some constants $c_1,c_2>0$. Then by Theorem 6 of \cite{chen2012} it holds
 \begin{equation}\label{weak:tang:cone:cond}
 \|\cT\phi-\cT\solq-\op_q(\phi-\solq)\|_W\leq c_1\, c_2\,\|\phi-\solq\|_Z^2.
\end{equation}
We further obtain for all $\phi\in\cB_{source,q}$ by making use of the Cauchy-Schwarz inequality
\begin{align*}
\|\phi-&\solq\|_Z^2=\sum_{j=1}^\infty \frac{s_{qj}}{s_{qj}}\Ex[(\phi(Z)-\sol(Z,q))e_j(Z)]^2\\
&\leq \Big(\sum_{j=1}^\infty s_{qj}^{-2}\Ex[(\phi(Z)-\sol(Z,q))e_j(Z)]^2\Big)^{1/2}
\Big(\sum_{j=1}^\infty s_{qj}^{2}\Ex[(\phi(Z)-\sol(Z,q))e_j(Z)]^2\Big)^{1/2}\\
&\leq c_0^{1/2}\,\|T_q(\phi-\solq)\|_W.
\end{align*}
Consequently, the tangential cone condition \eqref{tang:cone:cond} is satisfied if we assume $c_0^{1/2}\,c_1\, c_2<1$. We also note that for our test of exogeneity in Section \eqref{sec:exog} only the weaker condition \eqref{weak:tang:cone:cond} is required.
 $\hfill\square$
\end{example}

\begin{example}[Primitive Conditions for Assumption \ref{ass:A5:q}]\label{exmp:loc:cont}
Let $F_{Y|ZW}$ denote the cumulative distribution function of $Y$ given $(Z,W)$ and assume that it is Lipschitz continuous with constant $C_L>0$, that is, 
 $|F_{Y|ZW}(y)-F_{Y|ZW}(y')|\leq C_L|y-y'|$ for all $(y,y')$.
Due to Assumption \ref{ass:A0} the Sobolev space $W^{\alpha,p}$ can be embedded in $W^{\alpha, \infty}$ (cf. Theorem 6 of \cite{2003Adams}). In particular, the supremum norm
is bounded on $\cB$ and moreover, Assumption \ref{ass:A5:q} holds true.
Indeed, $\int_0^1\|\phi_q-\solq\|_\infty^2dq\leq r_n^2$ implies $\|\phi_q-\solq\|_\infty\leq c \,r_n$ for almost all $0<q<1$ and some constant $c>0$.
Hence, $\sol(Z,q)- c \,r_n\leq \phi(Z,q)\leq \sol(Z,q)+ c \,r_n$ for almost all $0<q<1$ and following \cite{CLVK03econometrics} (page 1599 -- 1600) we observe
\begin{align*}
 \Ex\Big[\int_0^1\max_{\phi\in\cB_{n}}\big(\1\{Y\leq\phi&(Z,q)\}-\1\{Y\leq \sol(Z,q)\}\big)^2dq\,\Big|W\Big]\\
  \leq& \int_0^1\Ex\Big[\1\big\{Y\leq\sol(Z,q)+ c \,r_n\big\}-\1\big\{Y\leq \sol(Z,q)- c \,r_n\big\}\Big|W\Big]dq\\\hfill
=&\int_0^1\Ex\Big[F_{Y|ZW}\big(\sol(Z,q)+ c \,r_n\big)-F_{Y|ZW}\big(\sol(Z,q)- c \,r_n\big)\Big|W\Big]dq\\
\leq& C_L\, c \,r_n
\end{align*}
which implies Assumption \ref{ass:A5:q} with $\kappa=1/2$.
$\hfill\square$
\end{example}

\begin{rem}[Local Overidentification]
In this remark, we discuss local overidentification restrictions in nonparametric instrumental quantile regression for some $0<q<1$.
 As \cite{chen2015overidentification}  point out in their Example 5.2, the range of the  Fr\'echet derivative  $T_q$, is  given by
   \begin{align*}
   \mathcal R_q =\set{\psi\in L_W^2:\, \psi=T_q\phi \text{ for some }\phi\in L_Z^2}.
   \end{align*}
   Local overidentification corresponds to the case where the closure of the range $\mathcal R_q$ is a strict subset of $L_W^2$. 
In this paper, the class of structural functions $\sol$ is restricted to belong to an ellipsoid $\cB$ and thus, we consider for each $q$:
   \begin{align*}
   \mathcal R_q(\cB) =\set{\psi\in L_W^2:\, \psi=T_q\phi \text{ for some }\phi\in \cB}.
   \end{align*}
 Mild restrictions on the ellipsoid $\cB$ imply local overidentification and hence, the class of functions in the alternative model is not empty. 
 $\hfill\square$
\end{rem}
The next result formalizes the discussion of the previous remark and shows that the regularity conditions imposed on the function set  $\mathcal B$ ensure overidentification.
\begin{prop}\label{prop:overident}
Let $\varPhi$ coincide with the Hilbert space $L_Z^2$ and let Assumption \ref{ass:A0} $(v)$ be satisfied. 
 Then we have local identification, i.e., for any $0<q<1$ the closure of $\mathcal R_q(\cB)$ is a strict subset of $L_W^2$. 
\end{prop}
The proof of Proposition \ref{prop:overident} relies on the fact that the functions in $\mathcal B$ are bounded by some constant $\rho>0$ and, in particular, no smoothness restrictions are employed here to achieve overidentification.\footnote{I thank an anonymous referee for suggesting this argumentation.} It is also possible to achieve overidentification for classes containing unbounded functions, as long as they satisfy minimal smoothness conditions.

The following result is due to \cite[Lemma 4.1]{chen2015overidentification} and gives a condition for local overidentification without imposing \textit{a priori} restrictions on the set of functions $\mathcal B$. 
\begin{lem}[\cite{chen2015overidentification}]\label{lemma:cs}
The model is locally overidentified if and only if
\begin{align*}
\left\{\psi\in L_W^2 : \EE[p_{Y|Z,W} (\varphi(Z,q),Z,W)\psi(W)|Z] = 0\right\} \neq \{0\}.
\end{align*}
\end{lem}
Lemma \ref{lemma:cs} provides a necessary and sufficient condition for local overidentification without imposing regularity or other shape restrictions. This result involves the adjoint of the Fr\'echet derivative $T_q$ and can be characterized more explicitly in different cases. 
For instance, assume that the vector of instruments can be decomposed such that $W=(W_1,W_2)$ with $p_{Y|Z,W} =p_{Y|Z,W_1}$, i.e., $W_2$ has no additional information on $Y$ which is not contained in $(Z,W_1)$. In this case, we have
\begin{align*}
\EE[p_{Y|Z,W} (\varphi(Z,q),Z,W)\psi(W)|Z]&=\EE\left[p_{Y|Z,W_1} (\varphi(Z,q),Z,W_1)\EE[\psi(W)|W_1,Z]\big|Z\right]
\end{align*}
 and hence, the model is locally overidentified when  there exists a nontrivial function $\psi$ such that $\EE[\psi(W)|Z,W_1]=0$. The last criterion is satisfied, for instance, if $W_2$ is independent of $(Z,W_1)$ for all $\psi$ which only depend on $W_2$ and $\EE[\psi(W_2)]=0$.

\paragraph{Notation}
For any $\phi\in\cB$ we introduce $\varPi_\k \phi\in\cB_\k$ satisfying $\|\varPi_\k\phi-\phi\|_{Z,p} = o(1)$. 
Further, we define  
\begin{equation*}
\omega_n=\max\Big(n^{-1}\l,\max_{\phi\in\cB_\k}\sum_{j>\l}\Ex[(\cT\phi(W)-q)f_j(W)]^2, \interleave T_\cdot(\varPi_\k\sol-\sol)\interleave_W^2\Big). 
\end{equation*}
The rate $\omega_n$ captures the variance and bias part for estimating $\mathcal T \phi$ for a fixed function $\phi$ and also contains the bias for approximating the structural function $\sol$ in the weak norm induced by the Fr\'echet derivative of $\mathcal T$.
Following \cite{chen08} we introduce the sieve measure of local ill-posedness by
\begin{equation*}
 \tau_\k:=\max_{\phi\in\cA_\k}\Big(\frac{\interleave\phi-\sol\interleave_{Z,p}^2}{\interleave T_{\cdot}(\phi-\sol)\interleave_W^2}\Big)
\end{equation*}
where $\cA_\k=\set{\phi\in\cB_\k^{(0,1)}:\interleave T_{\cdot}(\phi-\sol)\interleave_W^2>0}$.  We write $a_n\sim b_n$ when there exist
constants $c,c' > 0$ such that $c b_n\leq a_n\leq c' b_n$ for sufficiently large $n$.
\subsection{Asymptotic distribution under the null hypothesis} 
 The following theorem establishes asymptotic normality of the test statistic $S_n$ after standardization under the null hypothesis $H_0$. 
\begin{theo}\label{prop:main}
Let  Assumptions \ref{ass:A1:q}--\ref{ass:A0} be satisfied.  Assume that
\begin{equation}\label{cond:theo:norm}
m_n^{-1}=o(1),\quad \m=o(n^{1/2})
 \end{equation}
 and in addition
\begin{equation}\label{prop:chen:est:cond}
  n\omega_n=o(\sqrt{\m}) \text{ and } \interleave\varPi_\k\sol-\sol\interleave_{Z,p}^2+\tau_\k\omega_n=o\big(m_n^{-(1+\epsilon)/\kappa}\big)
\end{equation}
for some $\epsilon>0$.
Then we have under $H_0$
\begin{equation*}
 3\sqrt{5/\m}\big(S_n-\m/6\big)\stackrel{d}{\rightarrow} \cN(0,1).
\end{equation*}
\end{theo}
To motivate the constants in the sieve mean and variance, respectively, we observe
\begin{equation*}
 \int_0^1\Ex[(\1\{Y\leq \sol(Z,q)\}-q)^2|W]dq=\int_0^1 q(1-q)dq=1/6
\end{equation*}
and
\begin{multline*}
 \int_0^1 \big(\Ex[(\1\set{Y\leq \sol(Z,q)}-q)(\1\set{Y\leq \sol(Z,q')}-q')|W]\big)^2d(q,q')\\
 =2\int_0^1  (\min(q,q')-qq')^2d(q,q')=(3\sqrt{5})^{-2},
\end{multline*}
see also the proof of Lemma \ref{normal:lem:1}. The required rate imposed in \eqref{cond:theo:norm} on $\m$ is milder than the rate requirement $\m=o(n^{1/3})$ imposed by \cite{Breunig2012} in case of nonparametric instrumental mean regression. This is due to the fact that in the latter case we do not have a lower bound for the sieve standard deviation in general, while in case of quantile regression the sieve standard deviation is $\sqrt\m$ within a positive constant. This can be exploited to weaken rate restrictions on $\m$. Further, note that restriction \eqref{prop:chen:est:cond} implies $\k=o(\sqrt\m)$ (by using that $\l\leq \k$). This requirement  essentially determines the degree of overidentification required for inference.

The rate restriction  $\tau_\k\omega_n=o\big(m_n^{-(1+\epsilon)/\kappa}\big)$ imposed in condition \eqref{prop:chen:est:cond} implies that the dimension parameter $\m$ dominates the effect of estimation of the structural function.
Consequently, the asymptotic behavior of our test statistic is  not affected by the estimation of $\sol$, regardless of the underlying degree of ill-posedness. Note that this rate restriction can be ensured by choosing $\k$ relative to decay of the sieve measure of local ill-posedness, which is described in more detail in Example \ref{examp:class:smooth} below.
We illustrate below that condition \eqref{prop:chen:est:cond} is satisfied under common smoothness restrictions on $\sol$ and mapping requirements of the Fr\'echet derivative $T_q$.
\begin{rem}\label{rem:chen}
Consider the Hilbert space case $\varPhi=L_Z^2$ and let $\{e_j\}_{j\geq 1}$ be an orthonormal basis in $L_Z^2$. In this case, $\varPi_\k\phi=\sum_{j=1}^\k\Ex[\phi(Z)e_j(Z)]e_j$. Let us assume the following two conditions.
 \begin{enumerate}
  \item [(i)] \textit{Sieve approximation error:} $\|\varPi_\k\phi-\phi\|_Z=O(k_n^{-\alpha/d_z})$ for all $\phi\in\cB$.
  \item [(ii)] \textit{Link condition:} $\int_0^1\|T_q (\varPi_\k\phi-\phi)\|_W^2dq\leq \sum_{j\geq1}\opw_j\Ex[(\varPi_\k\phi-\phi)(Z)e_j(Z)]^2$  for all $\phi\in\cB$ and some positive nonincreasing sequence $(\opw_j)_{j\geq 1}$.
 \end{enumerate}
If the p.d.f. $p_Z$ of $Z\in[0,1]^{d_z}$ is bounded then it is well known that the sieve approximation error condition holds for splines, wavelets, and Fourier series bases. Due to Assumption \ref{ass:A0} $(v)$ the link condition is always satisfied with $\opw_j=1$ for all $j\geq 1$. The link condition implies an upper bound for the sieve measure of ill-posedness; that is, $\tau_\k\leq C\opw_\k$ for some constant $C>0$ and all $n\geq 1$ (cf. Lemma B.2 of \cite{chen08}). 
Consequently, the first part of condition \eqref{prop:chen:est:cond} simplifies to 
\begin{equation*}
  \max\big(\l,n\,l_n^{-2\beta/d_w}, n\opw_\k k_n^{-2\alpha/d_z}\big)=o(\sqrt{\m})
\end{equation*}
if $\{\cT\phi:\,\phi\in\cB_\k\}$ belongs to a H\"older space with H\"older parameter $\beta$. In addition, in the setting of Example \ref{exmp:loc:cont}, the second part of condition \eqref{prop:chen:est:cond} simplifies to 
\begin{equation*}
  m_n^{1+\epsilon}\max\big(n^{-1}\l,\,l_n^{-2\beta/d_w},k_n^{-2\alpha/d_z}\big)=o(1)
\end{equation*}
for some $\epsilon>0$. 
$\hfill\square$
\end{rem}
In the next example, we illustrate different mapping properties of the operator $T_q$ which are usually studied in the literature.
\begin{example}\label{examp:class:smooth}
Consider the Hilbert space setting of Remark \ref{rem:chen} with conditions $(i)$ and $(ii)$. In addition assume that the reverse link condition $\int_0^1\|T_q \phi\|_W^2dq\geq c\sum_{j\geq1}\opw_j\Ex[\phi(Z)e_j(Z)]^2$ for $\phi\in\cB$ and some constant $c>0$ is satisfied. In the setting of Example \ref{exmp:ident}, we have $\int_0^1s_{qj}^2dq>\opw_j$ for all $j\geq 1$ implying that $T_q$ is nonsingular for almost all $0<q<1$ (since any countable union of null sets is null).
For simplicity, let $Z$ and $W$ be scalars. Further, let $\max\big(n^{-1}\l, l_n^{-2\beta}\big)\sim n^{-1}\k$ and $\k\sim n^\chi$ for some constant $\chi>0$ which is specified in the following two cases.
\begin{enumerate}
 \item[(i)] \textit{Mildly ill-posed case:} If $\opw_\k\sim k_n^{-2\zeta}$ for some $\zeta\geq 0$ then in order for \eqref{prop:chen:est:cond} to hold we require $\m\sim n^\iota$ with $0<\iota<1/3$ and
\begin{equation*}
 (1-\iota/2)/(2\alpha+2\zeta)<\chi<\iota/2.
\end{equation*}
Further, $\int_0^1\|\varPi_\k\solq-\solq\|_{Z}^2dq+\tau_\k\omega_n=O(k_n^{-2\alpha}+k_n^{2\zeta+1}n^{-1})$ which is $o(m_n^{-2/\kappa})$ if
$\iota/(\alpha\kappa)<\chi<(1-2\iota/\kappa)/(2\zeta+1)$.
Thus, condition \eqref{prop:chen:est:cond} is satisfied if
\begin{equation*}
 \max\Big((1-\iota/2)/(2\alpha+2\zeta),\iota/(\alpha\kappa)\Big) <\chi<\min\Big(\iota/2, (1-2\iota/\kappa)/(2\zeta+1)\Big).
\end{equation*}
 \item[(ii)] \textit{Severely ill-posed case:} If $\opw_\k\sim \exp\big(-k_n^{2\zeta}\big)$ for some $\zeta>0$ then $\int_0^1\|\varPi_\k\solq-\solq\|_{Z}^2dq+\tau_\k\omega_n=O(k_n^{-2\alpha}+\exp(k_n^{2\zeta})k_nn^{-1})$. Thereby, 
 condition \eqref{prop:chen:est:cond} is satisfied if, for example, $\m=o\big((\log n)^{\alpha\kappa/\zeta}\big)$ and $\k\sim (\log n)^{1/\zeta}$.
\end{enumerate}
In both situations we conclude that the dimension parameter $\m$ is required to be larger than the dimension $\k$ of the sieve space for $n$ sufficiently large. Roughly speaking we require more moment restrictions implied by the instrument than the number of parameters we want to estimate. This corresponds to the test of overidentification in the parametric framework.
\hfill$\square$ 
\end{example}

In contrast to a test integrated over all quantiles, one might be interested to check model \eqref{model:NP} for one specific quantile. In this case, we consider the test statistic 
\begin{equation}\label{sn(q)}
 S_n(q)=\Big(\sum_{i=1}^n(\1\set{Y_i\leq \hsolq(Z_i)}-q)f_\umn(W_i)\Big)^t(\bold W_\m^t\bold W_\m)^-\sum_{i=1}^n(\1\set{Y_i\leq \hsolq(Z_i)}-q)f_\umn(W_i)
\end{equation}
If $S_n(q)$ becomes too large then we reject the null hypothesis $H_0$. 
The derivation of the asymptotic behavior of $S_n(q)$ is similar as in Theorem \ref{prop:main}. Indeed, only the Lebesgue measure over $(0,1)$ has to be replaced by the Dirac measure which has its mass at the quantile of interest.
\begin{coro}\label{coro:main}
Let  Assumptions \ref{ass:A1:q} and \ref{ass:A0} be satisfied.  For a fixed quantile $q\in(0,1)$, let Assumptions \ref{ass:A3:q}, \ref{ass:A5:q}, and conditions \eqref{cond:theo:norm} and \eqref{prop:chen:est:cond} hold.
If there exists a function $\sol_q\in\cB$ with $\cT\solq=q$ then
\begin{equation*}
 (2\m)^{-1/2}\Big(\frac{1}{q(1-q)}\, S_n(q)-\m\Big)\stackrel{d}{\rightarrow} \cN(0,1).
\end{equation*}
\end{coro}
In addition, one might be interested in certain regions of quantile functions. Let $\mu$ denote any measure on $(0,1)$. Again, the next result is a direct implication of Theorem \ref{prop:main} and hence we omit its proof. 
\begin{coro}\label{coro:main:int}
Let  Assumptions \ref{ass:A1:q} and \ref{ass:A0} be satisfied.   For all $q$ in the support of $\mu$, let Assumptions \ref{ass:A3:q}, \ref{ass:A5:q}, and conditions \eqref{cond:theo:norm} and \eqref{prop:chen:est:cond} hold.
If there exists a function $\sol\in\cB$ with $\int|\cT\solq-q|d\mu(q)=0$ then
\begin{equation*}
 \Big(2\m\int_0^1 (\min(q,q')-qq')^2d\mu(q)d\mu(q')\Big)^{-1/2}\Big( \int_0^1 S_n(q)d\mu(q)-\m\int_0^1 q(1-q)d\mu(q)\Big)\stackrel{d}{\rightarrow} \cN(0,1).
\end{equation*}
\end{coro}
As mentioned in the introduction, our test is a joint test of instrument validity and monotonicity of $\sol$ in its second entry. The following remark illustrates how the test statistic $S_n(q)$ integrated over a subset of $(0,1)$ can be useful to detect which kind of deviation exists. 
\begin{rem}[Detecting the kind of deviation]
Suppose that the structural function is strictly monotonically increasing in its second entry for values $q\in(0,q')$ given some $q'\in(0,1)$ (can be checked using Corollary \ref{coro:main:int}). Further, let $q\mapsto\sol(\cdot,q)$ be either nonincreasing or decreasing on $(q',q'')$. This can be assured by letting $q''$ close to $q'$ and assuming that $\sol$ does not oscillate for $q\geq q'$.
If $W$ is a valid instrument, employing model equation \eqref{model:NP:nonsep} and $V\sim\cU(0,1)$ yields
\begin{align*}
\PP(Y\leq \sol(Z,q)|W)&=\PP(\sol(Z,V)\leq \sol(Z,q)|W)\\
&\leq\PP(V\leq q|W)\\
&=q
\end{align*}
for all $q\leq q''$ and $q''$ sufficiently close to $q'$. 
The last inequality holds regardless whether the function $q\mapsto\sol(\cdot,q)$ is strictly monotone or not.
 Consequently, if $\inf_{w\in\cW}\PP(Y\leq \sol(Z,q)|W=w)>q$ for some $q\in(q',q'')$ we may conclude that $W$ is not a valid instrument. The analysis of a one sided test based on this inequality is beyond the scope of this paper. On the other hand, we can check the kind of deviation by using the estimator $\inf_{w\in\cW} f_\umn(w)^t\big[n^{-1}\sum_{i=1}^n(\1\set{Y_i\leq \hsolq(Z_i)}-q)f_\umn(W_i)\big]$. Further, confidence statements can be achieved by using resampling methods.
\hfill$\square$
\end{rem}

\begin{rem}[Implementation of the test statistic]
This remark provides some details on the implementation of our test. First, discretize the  $(0,1)$--integral by using the grid $1/N,2/N,\dots,(N-1)/N$ for some integer $N$. In different simulations, we found that a grid size of $N=20$ was sufficiently large. Also note that by the choice of the grid we avoid evaluation at boundary points zero or one. Second, for any integer $\m\leq n^{1/2}$ estimate the structural effect $\sol_q$ given in \eqref{def:est} for each grid point $q$, each parameter $\k$ with $k_n^2\leq \m$ and  $\l=2\k$. Third, compute the standardized test statistic $S_n$ such that it is maximized w.r.t.  $\m$ and minimized w.r.t. $\k$. That is, we choose $\k$ to provide a good model fit and $\m$ to increase the power of the test. 
The choice of the dimension  parameters capture essential rate requirements imposed to achieve asymptotic normality and is also motivated by simulation results. This leads to a so-called minimum-maximum principle, see also Subsection \ref{MC:sec:spec} for more details. 
\hfill$\square$
\end{rem}

\subsection{Consistency against a fixed alternative}
Let us first establish consistency when $H_0$ does not hold, that is, there exists no function $\sol$ belonging to $\cB^{(0,1)}$ which solves $\cT\solq=q$ for all $0<q<1$. The following proposition shows that our test has the ability to reject a false null hypothesis with probability $1$ as the sample size grows to infinity.
In the following analysis of the asymptotic power of our testing procedure we let $\solq=\argmin_{\phi\in\cB}\|\cT\phi-q\|_W$. So if $H_0$ is false then $\int_0^1\|\cT\solq-q\|_W^2dq>0$ since $p_W$ is uniformly bounded from below.
\begin{prop}\label{prop:cons:q}Assume that $H_0$ does not hold. Let  Assumptions \ref{ass:A1:q}--\ref{ass:A0} be satisfied. Consider a sequence $(\gamma_n)_{n\geq 1}$ satisfying $\gamma_n=o(n/\sqrt\m)$.
If conditions \eqref{cond:theo:norm} and \eqref{prop:chen:est:cond} hold  we have
\begin{gather*}
\PP\Big(3\sqrt{5/\m}\big(S_n-\m/6\big)>\gamma_n\Big)=1+o(1)\label{prop:cons:1}.
\end{gather*}
\end{prop}

\subsection{Limiting behavior under local alternatives}
In the following, we study the power of the test, that is, the probability to reject a false hypothesis against a sequence of linear local alternatives that tends to zero as the sample size tends to infinity.
We proceed similarly as \cite{Ait2001} (Section 3.3). More precisely, let $(\sol_{qn})_{n\geq 1}$ be a sequence of (nonstochastic) functions satisfying $n\int_0^1\|\cT\sol_{qn}-\cT\solq\|_W^2dq=o(\sqrt\m)$ where $\solq=\argmin_{\phi\in\cB}\|\cT\phi-q\|_W$. Then we consider alternative models defined by $\sol_{qn}$ with
\begin{equation}\label{loc:alt:ind:q}
\int_0^1\big\|\cT\sol_{qn}-q- \delta_n\xi_q\big\|_W^2dq=o(\delta_n^2)\quad  \text{ where } \quad \delta_n^2=\sqrt\m/(3\sqrt5 \, n).
\end{equation}
Here, $\xi_q\in L_W^2$ is a function satisfying $\int_0^1\|\xi_q\|_W^2dq>0$.
The next result establishes asymptotic normality for the standardized test statistic $S_n$.
\begin{prop}\label{coro:norm:q}
Let  Assumptions \ref{ass:A1:q}--\ref{ass:A0} be satisfied. Assume that $(\sol_{qn})_{n\geq 1}$ satisfies \eqref{loc:alt:ind:q} and $n\int_0^1\|\cT\sol_{qn}-\cT\solq\|_W^2dq=o(\sqrt\m)$.
If conditions \eqref{cond:theo:norm} and \eqref{prop:chen:est:cond} hold we have
\begin{equation*}
3\sqrt{5/\m}\big( S_n-\m/6\big)\stackrel{d}{\rightarrow}\cN\Big(\sum_{j= 1}^\infty\int_0^1\Ex[\xi_q(W)f_j(W)]^2dq,1\Big).
\end{equation*}
\end{prop}
  From Proposition \ref{coro:norm:q} we see that our test  can detect local linear alternatives at the rate $\delta_n$. If $\{f_j\}_{j\geq 1}$ forms an orthonormal basis in $L_W^2$ then $\delta_n$ coincides with $m_n^{1/4}n^{-1/2}$ within a constant. Hence, our test has the same power against local linear alternatives as the test of \cite{Hong95} who consider parametric specification testing.

\subsection{Inference based on bootstrap}\label{sub:sec:boot}
Nonparametric tests that rely on the asymptotic normal approximation may perform poorly in finite samples. An alternative approach is to use bootstrap approximation. It is known that bootstrap based procedures could approximate finite sample distributions more accurately. In the following, we propose a bootstrap version of our test statistic $S_n$. 

 The bootstrap procedure is based on a sequence of independent and identically distributed random variables $\varepsilon_i$, $1\leq i\leq n$, drawn independently of the original data $(Y_i,X_i,W_i)$, $1\leq i\leq n$.
Following \cite{chen2015} we then consider the bootstrap residual function
\begin{align*}
\varepsilon_{i}\big(\1{\{Y_i\leq \solq(Z_i)\}}-q\big).
\end{align*}
Let $\widehat \sol_{qn}^*$ be the bootstrap version of  the sieve least squares estimator \eqref{def:est}, which is computed in the same way but where only $\big(\1{\{Y_i\leq \phi(Z_i)\}}-q\big)$ is replaced by $\varepsilon_{i}\big(\1{\{Y_i\leq \phi(Z_i)\}}-q\big)$. The bootstrap version $S_n^*$ of our test statistic $S_n$ given in  \eqref{sn} builds on $\widehat \sol_{qn}^*$. 
More precisely, $S_n^*$ is computed as the test statistic $S_n$  but where only $\big(\1{\{Y_i\leq \hsolq(Z_i)\}}-q\big)$ is replaced by $\varepsilon_{i}\big(\1{\{Y_i\leq \widehat \sol_{qn}^*(Z_i)\}}-q\big)$. 
 
\begin{assA}\label{ass:B}
Let $(\varepsilon_i)_{i\geq 1}$ be an independent and identically distributed  sequence of random variables drawn independently of $(Y,Z,W)$ such that $\Ex[\varepsilon]=1$, $\Var(\varepsilon)=:\sigma_\varepsilon^2\in(0,\infty)$ and $\Ex[|\varepsilon-1|^4]<\infty$
\end{assA}
Assumption \ref{ass:B} corresponds to Assumption Boot.1 of \cite{chen2015}. We slightly strengthen their assumption by imposing a fourth moment restriction, which we require  to derive asymptotic validity of the bootstrap procedure. 
Due to the bootstrap innovations $\varepsilon_i$ the constants in the sieve mean and sieve standard deviation change.
 For the bootstrap test $S_n^*$ we obtain the sieve mean constant
\begin{equation*}
 \int_0^1\Ex[\varepsilon^2(\1\{Y\leq \sol(Z,q)\}-q)^2|W]dq=(\sigma_\varepsilon^2+1)/6
\end{equation*}
and the sieve standard deviation constant
\begin{equation*}
 \Big(\int_0^1 \big(\Ex[\varepsilon^2(\1\set{Y\leq \sol(Z,q)}-q)(\1\set{Y\leq \sol(Z,q')}-q')|W]\big)^2d(q,q')\Big)^{1/2}
 =(\sigma_\varepsilon^2+1)/(3\sqrt{5}).
\end{equation*}
\cite{chen2015} show that the bootstrap version of the sieve estimator $\widehat \sol_{qn}^*$ converges at the same rate as $\hsolq$. Thus, following line by line the proof of Theorem \ref{prop:main} and using the imposed restrictions on the weights $\varepsilon_i$ we obtain the following result.
 \begin{coro}\label{prop:main:B}
Let  the assumptions of Theorem \ref{prop:main:B} be satisfied.  
Under Assumption \ref{ass:B} and null hypothesis $H_0$ we have
\begin{equation*}
 3\sqrt{5/(\m(\sigma_\varepsilon^2+1))}\big(S_n^*-\m(\sigma_\varepsilon^2+1)/6\big)\stackrel{d}{\rightarrow} \cN(0,1).
\end{equation*}
\end{coro}
It should be emphasized the asymptotic validity of the bootstrap procedure is, in particular, due to the rate condition \eqref{prop:chen:est:cond}, which ensures that the asymptotic distribution of $S_n^*$ is not affected by the estimation of the structural function. The next result establishes consistency of the bootstrap test against fixed alternatives. 
 \begin{coro}\label{prop:main:B:power}
Assume that $H_0$ does not hold and that  the assumptions of Proposition \ref{prop:cons:q} are satisfied.  
Under Assumption \ref{ass:B} we have
\begin{gather*}
\PP\Big(3\sqrt{5/(\m(\sigma_\varepsilon^2+1))}\big(S_n^*-\m(\sigma_\varepsilon^2+1)/6\big)>\gamma_n\Big)=1+o(1).
\end{gather*}
\end{coro}

\section{Extensions}\label{sec:extension}
As we see in this section, our testing procedure can potentially be applied to a much wider range of situations. We now discuss corollaries that
generalize the previous results in different ways. For the following analysis we focus on a fixed quantile $q\in(0,1)$.

\subsection{Testing exogeneity}\label{sec:exog}
Falsely assuming exogeneity of the regressors leads to inconsistent estimators while on the other hand treating exogenous regressors as if they were endogenous can lower the rate of convergence dramatically. 
In this subsection, we develop a nonparametric test of exogeneity that is robust against possible nonseparability of unobservables. The test statistic is similar to the statistic $S_n(q)$ given in \eqref{sn(q)} but where $\hsolq$ is replaced by an estimator of the conditional quantile function. 

In contrast to the previous section, we assume here that there exists a unique function $\sol_q$ satisfying $Y= \solq(Z)+U_q$ with $\PP(U_q\leq 0|W)=q$ and for some $q\in(0,1)$. The relation between $Z$ and $W$ is thus restricted through this maintained hypothesis. 
Under the maintained hypothesis, we propose a test whether the vector of regressors $Z$ is exogenous at a quantile $q\in(0,1)$, that is, 
\begin{equation*}
H_0^{\textsl e}:\,\PP(U_q\leq 0|Z)=q.
\end{equation*}
In the following, we denote the conditional quantile function by $\sol^{\textsl e}_q$ which satisfies $\PP(Y\leq \sol^{\textsl e}_q(Z)|Z)=q$.
The null hypothesis $H_0^{\textsl e}$ is satisfied if and only if the structural function $\solq$ coincides with the conditional quantile function $\sol_q^{\textsl e}$. 
 Further, under nonsingularity of the operator $\cT$, hypothesis $H_0^{\textsl e}$ is equivalent to
\begin{equation}
 \cT\sol^{\textsl e}_q=q.
\end{equation}
Our test of exogeneity, which we propose below, is based on this equation or equivalently on $\PP(Y\leq \sol^{\textsl e}_q(Z)|W)=q$. More precisely, to test exogeneity we replace in the statistic $S_n(q)$ given in \eqref{sn(q)} the estimator of $\solq$ by an estimator of $\sol^{\textsl e}_q$.

In the following, $\hsol^{\textsl e}_{qn}$ denotes an estimator for the conditional quantile function $\sol^{\textsl e}_q$. For instance,  an estimator of $\sol^{\textsl e}_q$ is given by
\begin{equation}\label{est:ex}
 \hsol^{\textsl e}_{qn}
 =\argmin_{\phi\in\cB_\k}\sum_{i=1}^n\varrho_q\big(Y_i-\phi(Z_i)\big)
\end{equation}
where $\varrho_q(u)=|u|-(2q-1)u$ is the check function and here, $\cB_\k=\big\{\phi\in\cB:\,\phi(\cdot)=\sum_{j=1}^\k\beta_je_j(\cdot)\big\}$.
 For B-spline basis functions and an additional penalty this estimator was proposed by \cite{koenker1994}. In the following, let $p_Z$ and $p_{Z|W}$ denote the marginal density of $Z$ and the conditional density of $Z$ given $W$, respectively.
\begin{assA}\label{ass:ex}
 (i) There exists a function $\solq\in\cB$ such that $\cT\solq=q$. 
(ii) $p_{Y|Z,W}(\cdot,Z,W)$ is continuously differentiable, $|\partial p_{Y|Z,W}(\cdot,Z,W)/\partial y|\leq C$ and $p_{Z|W}(\cdot,W)\leq Cp_{Z}(\cdot)$ for some constant $C>0$. 
(iii) There exists a sequence $(R_n^{\textsl e})_{n\geq 1}$ with $R_n^{\textsl e}=o(1)$ such that $\|\hsol_{qn}^{\textsl e}-\sol_q^e\|_Z^2=O_p(R_n^{\textsl e})$.
\end{assA}
Assumption \ref{ass:ex} $(i)$ formalizes the maintained hypothesis of a correctly specified nonparametric instrumental quantile moment equation. Section \ref{sec:main} provides a test for it.
Due to Assumption \ref{ass:ex} $(ii) $ we do not require Assumption \ref{ass:A3:q} $(ii)$ but can rather rely on an upper bound of the Taylor reminder of $\cT$ obtained by \cite{chen2012}. In this sense, the test of exogeneity presented below requires weaker restrictions on the local curvature of $\cT$ than in the case of specification testing. Assumption \ref{ass:ex} specifies a rate requirement for the $L_Z^2$ distance of the estimator $\hsol_{qn}^{\textsl e}$. For instance, under $H_0^{\textsl e}$,  Assumption \ref{ass:ex} $(iii)$ is satisfied with $R_n^{\textsl e}=\k/n+k_n^{-2r}$ when $\hsol_{qn}^{\textsl e}$ is given by the estimator \eqref{est:ex} with the B-splines basis functions $\{e_j\}_{j\geq 1}$ and $Z$ is scalar, see \cite{he1994}. The same rate is obtained by \cite{horowitz2005} in the case of multivariate $Z$ in an additive quantile regression model.

For a test of the null hypothesis  $H_0^{\textsl e}$ we replace in the definition of $S_n(q)$ given in \eqref{sn(q)} the estimator $\hsolq$ by $\hsol^{\textsl e}_{qn}$. That is,
\begin{equation*}
 S_n^{\textsl e}(q)=\Big(\sum_{i=1}^n(\1\{Y_i\leq \hsolq^{\textsl e}(Z_i)\}-q)f_\umn(W_i)\Big)^t(\bold W_\m^t\bold W_\m)^-\sum_{i=1}^n(\1\{Y_i\leq \hsolq^{\textsl e}(Z_i)\}-q)f_\umn(W_i)
\end{equation*}
We reject the hypothesis $H_0^{\textsl e}$ if $S_n^{\textsl e}(q)$ becomes too large.
The next result establishes asymptotic normality of our test statistic $S_n^{\textsl e}(q)$ under the null hypothesis.
\begin{coro}\label{coro:int}
Let Assumptions \ref{ass:A1:q}, \ref{ass:A3:q} $(i)$, \ref{ass:A5:q}, \ref{ass:A0} and \ref{ass:ex} hold. Let $\m$ satisfy condition \eqref{cond:theo:norm}. Consider the estimator $\hsolq^{\textsl e}$ given in \eqref{est:ex} where $\k$ satisfies
\begin{equation}\label{cond:ex}
n R_n^{\textsl e}=o(\sqrt{\m})\quad 
\text{and}\quad R_n^{\textsl e}=o\big(m_n^{-(1+\epsilon)/\kappa}\big)
\end{equation}
for some $\epsilon>0$.
 Then we have under $H_0^{\textsl e}$
\begin{equation*}
 \big(2\m\big)^{-1/2}\Big(\frac{1}{q(1-q)}\,S_n^{\textsl e}(q)-\m\Big)\stackrel{d}{\rightarrow} \cN(0,1).
\end{equation*}
\end{coro}

\begin{example}
Let us illustrate when condition \eqref{cond:ex} holds true.
Let $\m\sim n^\iota$ with $0<\iota<1/3$. Then for \eqref{cond:ex} to hold let $\k\sim n^\chi$ where $\chi>0$ satisfies
\begin{equation*}
 \max\Big(\frac{1-\iota/2}{2r},\,\frac{\iota}{r\kappa}\Big)<\chi<\min\Big(\frac{\iota}{2}, 1-\frac{2\iota}{\kappa}\Big).
\end{equation*}
Hence, we require $r>2/\kappa$ which is a slightly stronger restriction than Assumption \ref{ass:A0} $(i)$. 
$\hfill\square$
\end{example}

In the following, we study the power of the test, that is, the probability to reject a false hypothesis against a sequence of linear local alternatives that tends to zero as the sample size tends to infinity.
More precisely, let $(\sol_{qn}^{\textsl e})_{n\geq 1}$ be a sequence of (nonstochastic) functions satisfying
\begin{equation}\label{loc:alt:ind:q:ex}
\big\|\cT\sol_{qn}^{\textsl e}-q- \delta_n\xi_q^{\textsl e}\big\|_W^2=o(\delta_n^2)\quad  \text{ where } \quad \delta_n^2=\sqrt{2\m}.
\end{equation}
Here, $\xi_q^{\textsl e}\in L_W^2$ is a function satisfying $\|\xi_q^{\textsl e}\|_W^2>0$.
The next result establishes asymptotic normality for the standardized test statistic $S_n^{\textsl e}(q)$.
\begin{coro}\label{loc:power:ex}
Let Assumptions \ref{ass:A1:q}, \ref{ass:A3:q} $(i)$, \ref{ass:A5:q}, \ref{ass:A0}, and \ref{ass:ex} be satisfied. Assume that $(\sol_{qn}^e)_{n\geq 1}$ satisfies \eqref{loc:alt:ind:q:ex}.
If condition \eqref{cond:ex} holds true we have
\begin{equation*}
\big(2\m\big)^{-1/2}\Big(\frac{1}{q(1-q)}\,S_n^{\textsl e}(q)-\m\Big)\stackrel{d}{\rightarrow}\cN\Big(\sum_{j= 1}^\infty\int_0^1\Ex[\xi_q^{\textsl e}(W)f_j(W)]^2dq,1\Big).
\end{equation*}
\end{coro}

\subsection{Testing additivity}\label{sec:add}
The test statistic given in \eqref{sn} is also convenient to check additional restrictions on the structural effect $\solq$ for $0<q<1$. These additional restrictions can be easily imposed by constraints on the functions of the sieve space $\cB_\k$. For instance, one may impose an additive structure of the quantile structural effects.

By assuming an additive structure of $\solq$ one might reduce the effect of dimensionality of the regressors on the convergence rate of an estimator (cf. \cite{chen08} in case of instrumental quantile regression).
Applying this structure leads, however, to inconsistent estimators in general if the function $\solq$ does not obey an additive form. Our aim in the following is to test whether
\begin{equation*}\label{H0:add}
 H_0^{add}:\text{there exist functions }\sol^1_q,\sol^2_q\in\cB\text{ such that }\,\PP(Y\leq \sol^1_q(Z')+\sol^2_q(Z'')|W)=q.
\end{equation*}
Similarly as above we obtain the test statistic
\begin{equation*}
  S_n^{add}(q)=\Big(\sum_{i=1}^n(\1\{Y_i\leq \hsolq^{\textsl add}(Z_i)\}-q)f_\umn(W_i)\Big)^t(\bold W_\m^t\bold W_\m)^-\sum_{i=1}^n(\1\{Y_i\leq \hsolq^{\textsl add}(Z_i)\}-q)f_\umn(W_i)
\end{equation*}
Here the estimator $\hsolq^{\textsl add}=(\hsol_{qn}^1,\hsol_{qn}^2)$ of $\solq=(\sol^1_q,\sol^2_q)$ is given by \eqref{def:est} where the sieve basis is a tensor product of basis functions that depend either on $Z'$ or $Z''$.  For a more detailed discussion we refer to Section 6 of \cite{chen08}. The next asymptotic normality result is a direct consequence of Corollary \ref{coro:main} and hence its proof is omitted.
\begin{coro}Given the conditions of Corollary \ref{coro:main} we have under $H_0^{add}$
\begin{equation*}
 \big(2\m\big)^{-1/2}\Big(\frac{1}{q(1-q)}\,S_n^{add}(q)-\m\Big)\stackrel{d}{\rightarrow} \cN(0,1).
\end{equation*}
\end{coro}

\section{Monte Carlo simulation}\label{sec:MC}
In this section, we study the finite sample performance of our test by presenting the results of a Monte Carlo investigation. There are $1000 $ Monte Carlo replications in each experiment. Results are presented for the nominal levels $0.05$. Let $\Phi$ denote the cumulative standard normal distribution function.
Throughout this simulation study, realizations $(Z,W)$ were generated by $Z = \Phi\big(\zeta\omega +\sqrt{1-\zeta^2}\,\varepsilon\big)$ and $W =\Phi(\omega)$ where  $\omega$ is independent of $\varepsilon$ and $\omega,\,\varepsilon\sim \cN(0, 1)$. Here, the constant $\zeta>0$ determines the degree of correlation between $Z$ and $W$ and is varied in the experiments. 

\subsection{Testing a Nonparametric Specification}\label{MC:sec:spec}
We begin with the finite sample analysis of our test statistics in case of nonparametric specification testing.
To analyze the finite sample power we distinguish in the following between a failure of the null hypothesis caused either by a lack of instrument validity or by non-monotonicity of the structural function in unobservables.
\paragraph{Failure of instrument validity.}
We first generate realizations of $Y$ under the null hypothesis $H_0$. Recall that under $H_0$ there exists a function $\sol\in\cB^{(0,1)} \text{ such that } \PP(Y\leq \sol(Z,q)|W)=q \text{ for all }q\in(0,1)$. In the following finite sample analysis, we restrict $\cB^{(0,1)}$ to contain continuously differentiable functions only. 
Under $H_0$ we generate realizations of $Y$ from the nonseparable model
\begin{equation}\label{sim:model}
Y = \phi(Z)(1+V/6) + V/2
\end{equation}
where $V = \vartheta\,\varepsilon + \sqrt{1-\vartheta^2}\,\epsilon$ with $\epsilon\sim \cN(0, 1)$ independent of $(\omega,\varepsilon)$ and  $\vartheta=0.7$.
We consider the function  $\phi(z)=\sum_{j=1}^\infty \,j^{-4} \cos(j\pi z)$.   For computational reasons we truncate the infinite sum at $100$.
 The resulting function is displayed in Figure \ref{sol.pic}. 
Since $\phi$ is continuously differentiable the null hypothesis $H_0$ is satisfied with $\sol(z,q)=\phi(z)\big(1+F_V^{-1}(q)/6\big) + F_V^{-1}(q)/2$, where $F_V^{-1}$ denotes the quantile function of $V$. 
\begin{figure}[h]
	\centering
	\caption{Graphs of $\phi$ and $\sol^{\textsl e}$}
		\includegraphics[width=10cm, height=6cm]{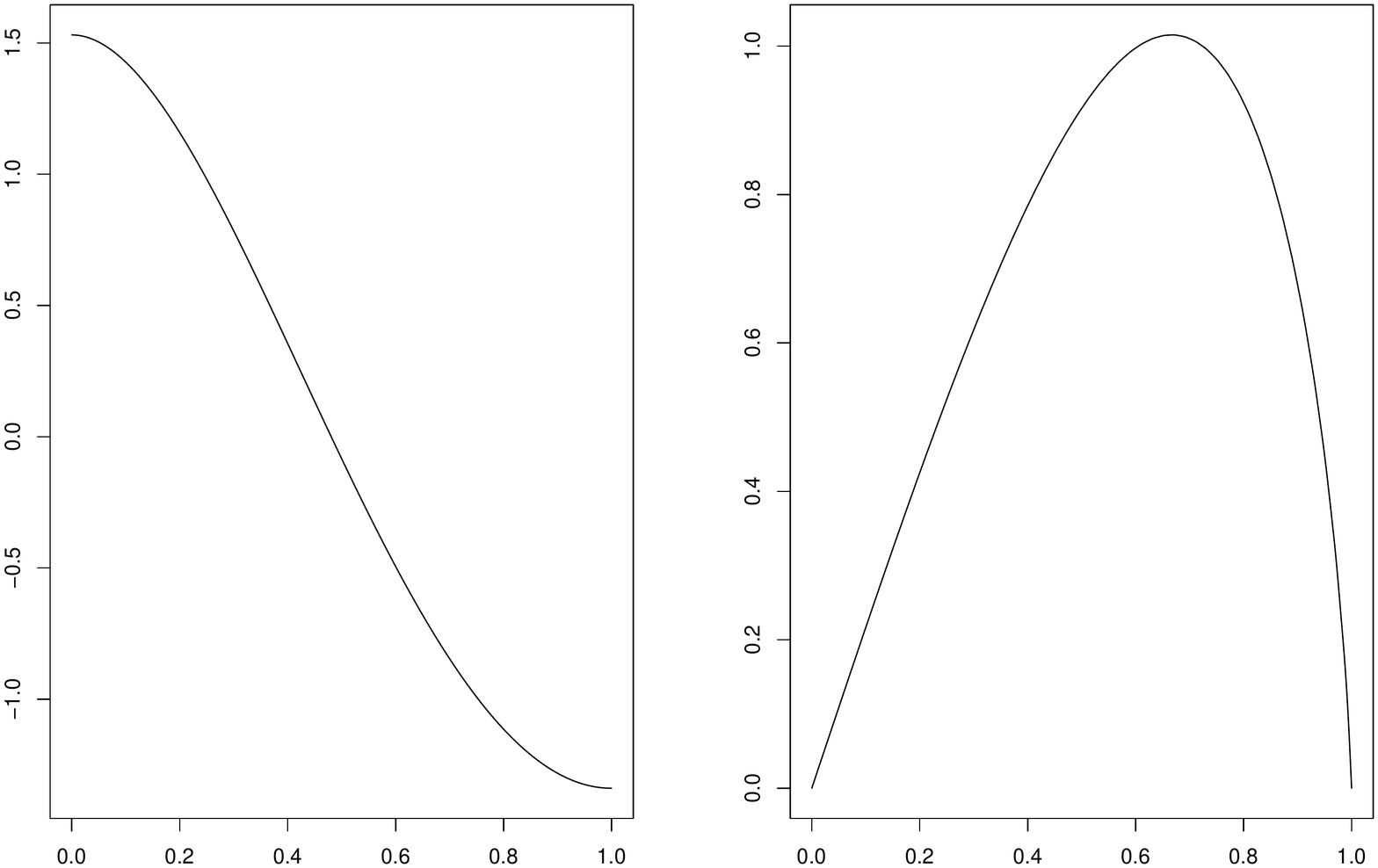}
	\label{sol.pic}
\end{figure}

\begin{table}[h]
\caption{Empirical rejection probabilities for the standardized test statistic $3\sqrt{5/\m}\big(S_n-\m/6\big)$ and its bootstrap version $3\sqrt{5/(\m(\sigma_\varepsilon^2+1))}\big(S_n^*-\m(\sigma_\varepsilon^2+1)/6\big)$ with varying dimension parameters $\k$ and $\m$ with $\l=2\k$.}
\begin{center}
\renewcommand{\arraystretch}{1.1}
\begin{tabular}{c|c|| c | c | c}
\textit{Sample} &\textit{Model}& &\textit{Emp. rejection prob.}& \textit{Emp. rejection prob.}\\
\textit{Size} && & \textit{ \quad using $S_n$ }& \textit{ \quad using $S_n^*$ }\\\hline	\hline
  &  &\diaghead{AAAAAA}{$k_n$}{$\m$} & 	\quad	 \textit{$20$\quad\quad$25$\quad\quad$30$}& 	\quad	 \textit{$20$\quad\quad$25$\quad\quad$30$}\\
\hline
500	&{\small $H_0$ true}&$4$ &$\mathbf{0.085}$\quad0.083\quad0.082 &$\mathbf{0.064}$\quad0.052\quad0.050\\
	  	&$\rho_1$&  						&$\mathbf{0.317}$\quad0.289\quad0.259 &$\mathbf{0.252}$\quad0.224\quad0.196\\
		&$\rho_2$&  						&$\mathbf{0.337}$\quad0.302\quad0.289 &$\mathbf{0.298}$\quad0.248\quad0.215\\
		&$\rho_3$&  						&$\mathbf{0.393}$\quad0.354\quad0.341 &$\mathbf{0.356}$\quad0.308\quad0.301\\
		&$\rho_3$&  						&$\mathbf{0.739}$\quad0.701\quad0.670 &$\mathbf{0.748}$\quad0.680\quad0.658\\
		&{\small $H_0$ true}&$5$	&0.076\quad0.076\quad0.080 &0.044\quad0.032\quad0.048\\
		&$\rho_1$&  						&0.195\quad0.179\quad0.169 &0.106\quad0.106\quad0.047\\
		&$\rho_2$&  						&0.200\quad0.194\quad0.174 &0.130\quad0.116\quad0.064\\
		&$\rho_3$&  						&0.171\quad0.153\quad0.152 &0.082\quad0.082\quad0.064\\
		&$\rho_3$&  						&0.270\quad0.257\quad0.228 &0.168\quad0.140\quad0.095\\
\hline
1000&{\small $H_0$ true}&$4$&0.077\quad$\mathbf{0.082}$\quad0.081 	&0.060\quad0.074\quad$\mathbf{0.076}$\\
	  	&$\rho_1$&  						&0.630\quad0.587\quad0.553  					&0.576\quad0.540\quad0.502\\
		&$\rho_2$&  						&0.636\quad0.582\quad0.549 						&0.576\quad0.544\quad0.492\\
		&$\rho_3$&  						&0.738\quad0.697\quad0.670 						&0.710\quad0.662\quad0.638\\
		&$\rho_3$&  						&0.905\quad0.882\quad0.864  					&0.938\quad0.924\quad0.896\\
		&{\small $H_0$ true}&$5$	&0.203\quad0.192\quad0.178 						&0.098\quad0.104\quad0.094\\
		&$\rho_1$&  						&0.554\quad$\mathbf{0.518}$\quad0.495 	&0.420\quad$\mathbf{0.396}$\quad0.380\\
		&$\rho_2$&  						&0.629\quad$\mathbf{0.596}$\quad0.549 	&0.532\quad$\mathbf{0.478}$\quad0.460\\
		&$\rho_3$&  						&0.423\quad$\mathbf{0.410}$\quad0.385 	&0.338\quad$\mathbf{0.314}$\quad0.272\\
		&$\rho_3$&  						&0.622\quad$\mathbf{0.593}$\quad0.574 	&0.576\quad$\mathbf{0.550}$\quad0.520\\
\end{tabular}
\end{center}
\label{table:instr}
\end{table}

When $H_0$ is false we generate realizations of $Y$ from
\begin{equation}\label{sim:model:alt}
Y = (\phi(Z)+\rho_j(Z))(1+V/6) + V/2
\end{equation}
 where $\rho_j(z)=10\,j\,(z\1\{z\leq 0.25\}+(z-1)\1\{z>0.25\})$
 for $j=1,2$ and  $\rho_j(z)=(z/2c_j)\1\{0.5-c_j\leq z<0.5+c_j\}$ for $j=3,4$, with $c_3=0.1$ and $c_4=0.05$. Here, the variable $V$ is generated as in \eqref{sim:model}. 
Under \eqref{sim:model:alt}, the structural function $\sol$ satisfying the
quantile restriction $\PP(Y\leq \sol(Z,q)|W)=q$ is given by $\sol(z,q)=(\phi(z)+\rho_j(z))(1+F_V^{-1}(q)/6) + F_V^{-1}(q)/2$. So $\sol(\cdot,q)$ is not continuously differentiable and thus, $H_0$ is false.
Due to the ill-posed inverse problem estimation of $\sol(\cdot,q)$ we cannot choose $\k$ sufficiently large to capture such irregularities which implies finite sample power of our test against those alternatives.
This corresponds  to the analysis of \cite{Horowitz2011} in the instrumental mean regression case.

 For each quantile $0<q<1$, we estimate the structural function using the estimator $\hsol_{qn}$ given in \eqref{def:est} with B-splines as approximation basis functions.
More precisely,  for the sieve space $\cB_\k$ we use B-splines of order 2 with 1 knot or 2 knots (hence $\k=4$ or $\k=5$) and for the criterion function we use B-splines of order 2 with 5 knots or 7 knots (hence $\l=2\k$), respectively. We thus follow \cite{chen2015optimal} and choose $\l$ to be a constant multiple of $\k$. Also for the vector of basis functions $f_\umn$, used to construct the test statistic, we use B-spline basis of order 2 with knots varying between 17, 22 or 27 (hence $\m=20$, $\m=25$ or $\m=30$).

The empirical rejection probabilities of our standardized test statistic $3\sqrt{5/\m}\big(S_n-\m/6\big)$ at nominal level $0.05$ are shown in Table \ref{table:instr}. 
We approximate the integral over the quantiles on $(0,1)$ by the mean of a random sample from the uniform $(0,1)$ distribution.
As we see from Table \ref{table:instr}, our test is less sensitive with respect to the choice of $\m$ than to the choice of $\k$, which is not surprising and well known from nonparametric instrumental variable estimation problems, see also \cite{chen2015}. Table \ref{table:instr} shows the empirical rejection probabilities for the sample sizes $500$ and $1000$. We see that as the sample size increases the finite sample rejection probabilities become larger in the alternative models. For $\k=4$ we see that the finite sample coverage improves slightly as the sample size increases. This is not the case for $\k=5$ which appears to be an inappropriate choice implying a large variance.
 
In Table \ref{table:instr} we also compare our testing procedure to a bootstrap version of it. We consider the generalized residual bootstrap as proposed in Subsection \ref{sub:sec:boot}. We generate the bootstrap weights by $\varepsilon\sim\mathcal N(1,\sigma_\varepsilon^2)$, independently of $(Y,X,W)$, where $\sigma_\varepsilon=0.5$.  We run $200$ bootstrap evaluations per Monte Carlo replication. We see from Table \ref{table:instr} that the bootstrap leads to an improvement in the finite sample coverage in the true model. In this sense, the bootstrap test statistic is less sensitive to the choice of $\k$ under the true model. Similar to \cite{chen2015} (see p. 1059), we see only a minor improvement of the bootstrap test in the alternative models but we expect that it improves further as the number of bootstrap runs is increased.

As we fix the dimension parameter $\l=2\k$, two dimension parameters remain to be chosen by the econometrician, namely, $\k$ and $\m$. 
While proposing an adaptive testing procedure is beyond the scope of this paper, we want to provide an heuristic argument for the parameter choice. 
Intuitively, we want to choose $\k$ such that we have a good model fit, i.e., a small value of the test statistic, and $\m$ to have good power properties, i.e., a large value of the test statistics. Moreover, the choice should reflect the rate requirement from our theory, that is,  $\k\leq \l =o(m_n^{1/2})$ and $m_n=o(n^{1/2})$.
We  implement such a heuristic  parameter choice criterion via the following minimum-maximum principle. That is, if $\set{s(\k,\m)}$ denotes the standardized value of our test $S_n$ with dimension parameters $\k$ and $\m$ then we choose these parameters such that
\begin{align*}
\min_{\k< n^{1/4}}\, \max_{k_n^2\leq m_n<n^{1/2}}\set{s(\k,\m)}.
\end{align*}
The values of this minimum-maximum principle (over the range $\m\in\set{20,25,30}$ and $\k\in\set{4,5}$) are shown in bold in Table  \ref{table:instr}.
Note that the requirement $\k<n^{1/4}$ implies  $\k\leq 4$ when $n=500$ and $\k\leq 5$ when $n=1000$.  Further, $m_n<n^{1/2}$ implies $m_n\leq 22$ for $n=500$ and $m_n\leq 31$ for $n=1000$. 
 We see that this criterion helps to avoid choosing the dimension parameter $\k$ too large which would yield inaccurate coverage.
Such a rule, however, does not account for ill-posedness of the estimation problem and hence, $\k$ might still be chosen too large. We thus could calculate the sieve measure of ill-posedness by estimating the first $\k$ minimal eigenvalues of $T_q$ (see also \cite{chen2015}). 
\paragraph{Failure of monotonicity in unobservables.}
We study the finite sample power of our test when $\sol$ is not strictly monotonic in the structural disturbance $V$. Realizations of $Y$ were generated from
\begin{equation}\label{mont:nonsep}
Y = \Phi(Z+V)V^2
\end{equation}
where $V = \Phi\big((\vartheta\,\varepsilon + \sqrt{1-\vartheta^2}\,\epsilon)/4\big)$ with $\epsilon\sim \cN(0, 1)$ and where $\vartheta=0.8$.
\begin{table}[h]
\caption{Empirical rejection probabilities for the standardized test statistic $3\sqrt{5/\m}\big(S_n-\m/6\big)$ and $3\sqrt{5/(\m(\sigma_\varepsilon^2+1))}\big(S_n^*-\m(\sigma_\varepsilon^2+1)/6\big)$ using varying dimension parameters $\k$ and $\m$ with $\l=2\k$.}
\begin{center}
\renewcommand{\arraystretch}{1.1}
\begin{tabular}{c|c|| c | c | c}
\textit{Sample} &\textit{Model}& &\textit{Emp. rejection prob.}& \textit{Emp. rejection prob.}\\
\textit{Size} && & \textit{ \quad using $S_n$ }& \textit{ \quad using $S_n^*$ }\\\hline	\hline
  &  &\diaghead{AAAAAA}{$k_n$}{$\m$} &\textit{$20$\quad\quad\,\,$25$\quad\quad\,\,$30$}& 	\textit{$20$\quad\quad\,$25$\quad\quad$30$}\\
\hline
500&{\small\eqref{mont:nonsep}}&$4$ 		&$\mathbf{0.043}$\quad0.066\quad0.079 	&$\mathbf{0.022}$\quad0.044\quad0.058\\
&{\small\eqref{mont:nonsep:alt:1} with j=1}&&$\mathbf{0.390}$\quad0.433\quad0.393 &$\mathbf{0.338}$\quad0.324\quad0.298\\
&{\small\eqref{mont:nonsep:alt:1} with j=2}&&$\mathbf{0.966}$\quad0.967\quad0.959 &$\mathbf{0.984}$\quad0.970\quad0.964\\
&{\small\eqref{mont:nonsep:alt:2} with j=1}&&$\mathbf{0.441}$\quad0.492\quad0.435 &$\mathbf{0.376}$\quad0.372\quad0.342\\
&{\small\eqref{mont:nonsep:alt:2} with j=2}&&$\mathbf{0.976}$\quad0.979\quad0.968 &$\mathbf{0.994}$\quad0.982\quad0.978\\
&{\small\eqref{mont:nonsep}}&$5$				&0.048\quad0.063\quad0.083 &0.024\quad0.030\quad0.036\\
&{\small\eqref{mont:nonsep:alt:1} with j=1}&&0.183\quad0.247\quad0.215 &0.132\quad0.126\quad0.110\\
&{\small\eqref{mont:nonsep:alt:1} with j=2}&&0.671\quad0.710\quad0.649 &0.722\quad0.662\quad0.602\\
&{\small\eqref{mont:nonsep:alt:2} with j=1}&&0.219\quad0.278\quad0.259 &0.154\quad0.144\quad0.112\\
&{\small\eqref{mont:nonsep:alt:2} with j=2}&&0.721\quad0.746\quad0.672 &0.766\quad0.704\quad0.650\\
\hline
1000&{\small\eqref{mont:nonsep}}&$4$		&0.042\quad0.080\quad0.082 &0.032\quad0.037\quad$\mathbf{ 0.038}$\\
&{\small\eqref{mont:nonsep:alt:1} with j=1}&&0.717\quad0.712\quad0.681 &0.696\quad0.677\quad0.636\\
&{\small\eqref{mont:nonsep:alt:1} with j=2}&&1.000\quad1.000\quad0.999 &1.000\quad1.000\quad1.000\\
&{\small\eqref{mont:nonsep:alt:2} with j=1}&&0.751\quad0.768\quad0.737 &0.752\quad0.733\quad0.694\\
&{\small\eqref{mont:nonsep:alt:2} with j=2}&&1.000\quad0.999\quad0.999 &1.000\quad1.000 \quad1.000\\
&{\small\eqref{mont:nonsep}}&$5$				 &0.044\quad0.055\quad$\mathbf{0.057 }$ &0.030\quad0.030\quad0.042\\
&{\small\eqref{mont:nonsep:alt:1} with j=1}&&0.452\quad$\mathbf{0.435}$\quad0.394 &0.414\quad$\mathbf{0.368}$\quad0.332\\
&{\small\eqref{mont:nonsep:alt:1} with j=2}&&0.966\quad$\mathbf{0.953}$\quad0.932 &0.982\quad$\mathbf{0.974}$\quad0.968\\
&{\small\eqref{mont:nonsep:alt:2} with j=1}&&0.515\quad$\mathbf{0.490}$\quad0.441 &0.490\quad$\mathbf{0.442}$\quad0.400\\
&{\small\eqref{mont:nonsep:alt:2} with j=2}&&0.971\quad$\mathbf{0.961}$\quad0.950 &0.984\quad$\mathbf{0.982}$\quad0.982\\
\end{tabular}
\end{center}
\label{table:mon}
\end{table} 
When $H_0$ is false we generate
\begin{equation}\label{mont:nonsep:alt:1}
Y = \Phi(Z+V)(V-0.5)^{2j}
\end{equation}
or
\begin{equation}\label{mont:nonsep:alt:2}
Y = \Phi(Z+V)\Phi^{-2j}(V)
\end{equation}
for $j=1,2$. In the alternative models, the structural disturbance enters the model in a nonmonotonic way. 
We construct the statistic $S_n$ and its bootstrap counterpart $S_n^*$ as described in the previous paragraph. 

Table \ref{table:mon} depicts the empirical rejection probabilities of our test against the alternative models \eqref{mont:nonsep:alt:1} and \eqref{mont:nonsep:alt:2}.
Again we observe that our test is not very sensitive to the choice of the dimension parameter $\m$. Our test becomes somewhat less powerful for large $\k$. But in contrast to the alternatives involving discontinuous functions in the previous paragraph,  the choice of $\k$ is not as sensitive. 
For each choice of parameter $\k$, our test becomes more powerful as the sample size increases from $500$ to $1000$. For $n=1000$ we see that the parameter choice $\k=5$ leads to a more accurate finite sample coverage. This is captured   by the minimum-maximum principle as introduced above. Again, the resulting values of the test statistic using this criterion  over the range $\m\in\set{20,25,30}$ and $\k\in\set{4,5}$ are shown in bold.
Again we observe that the boostrap version of the test statistic behaves similarly as the statistic $S_n$.

\subsection{Testing exogeneity}\label{MC:sec:exog}
Realizations $Y$ were generated by
\begin{equation*}
Y = \sol^{\textsl e}(Z)+ V/2
\end{equation*}
where $V$ is generated as described in model \eqref{sim:model}, that is, $V = \vartheta\,\varepsilon + \sqrt{1-\vartheta^2}\,\epsilon$ with $\epsilon\sim \cN(0, 1)$ independent of $(\omega,\varepsilon)$.
The function $\sol^{\textsl e}$ is given by $\sol^{\textsl e}(z)=\sum_{j=1}^\infty (-1)^{j+1}\,j^{-2} \sin(j\pi z)$. Again, for computational reasons we truncate the infinite sum at $100$. 
 The resulting function is displayed in Figure \ref{sol.pic}.
Note that $\vartheta$ determines the degree of endogeneity of $Z$ and is varied among the experiments. The null hypothesis $H_0:\PP(Y\leq \sol^{\textsl e}(Z)|Z)=q$ holds true if $\vartheta=0$ and is false otherwise. In the following, we perform a test at the median $q=0.5$. As our test relies on the equation $\PP(Y\leq \sol^{\textsl e}(Z)|W)=q$ we expect our test to have more power as the correlation between $W$ and $Z$ increases. 

The test statistic is implemented as described in Section \ref{MC:sec:exog}. To estimate the structural effect we make use of the estimator $\hsol^{\textsl e}_{qn}$ of \cite{he1994} given in 
\eqref{est:ex}. Here, we use B-splines of order 2 with 1 knot (hence $\k=4$) or 2 knots (hence $\k=5$). In contrast to the previous section, the choice of the dimension parameter $\k$ is not affected by the ill-posedness of the underlying inverse problem. As above, the vector of basis functions $f_\umn$ is also constructed with B-spline basis of order 2 with knots varying between 17, 22 or 27 (hence $\m=20$, $\m=25$ or $\m=30$). 

Table \ref{table:ex}	 depicts the empirical rejection probabilities with varying number of basis functions. As we see from Table \ref{table:ex}, our test becomes more powerful for larger $\zeta$; that is, for instruments with a stronger correlation to the covariates $Z$. From Table \ref{table:ex} we see that the test of exogeneity becomes somewhat less powerful for larger values of $\m$. On the other hand, the test seems not to be too sensitive with respect to the choice of the dimension parameters $\k$ and $\m$.  We also see from Table \ref{table:ex} that the finite sample coverage and power properties of the test improve as the sample size increases from $500$ to $1000$. 

Similarly as above, a guideline for smoothing parameter choice in practice is given by the following minimum-maximum principle. That is, if $\set{s_q^{\textsl e}(\k,\m)}$ denotes the standardized value of our test $S_n^{\textsl e}(q)$ with dimension parameters $\k$ and $\m$ then choose these parameters such that
\begin{align*}
\min_{\k< n^{1/4}}\, \max_{k_n^2\leq m_n<n^{1/2}}\set{s_q^{\textsl e}(\k,\m)}.
\end{align*}
Again this criterion takes the rate condition for the asymptotic theory into account. In Table \ref{table:ex} the resulting values of the test statistic using this criterion  over the range $\m\in\set{20,25,30}$ and $\k\in\set{4,5}$ are shown in bold.
\begin{table}[h]
\caption{Empirical rejection probabilities for the standardized test statistic $(2\m)^{-1/2}\big(4\,S_n^{\textsl e}(0.5)-\m\big)$ with varying dimension parameters $\k$ and $\m$.}
\begin{center}
\renewcommand{\arraystretch}{1.1}
\begin{tabular}{c|c|| c | c | c}
\textit{ $\zeta$}&\textit{ $\vartheta$ }& &\textit{ Emp. rejection prob.}& \textit{ Emp. rejection prob.}\\
& && \textit{ \quad using $S_n^{\textsl e}(0.5)$ with $n=500$}& \textit{ \quad using $S_n^{\textsl e}(0.5)$ with $n=1000$}\\\hline	\hline
  &&\diaghead{AAAAAA}{$k_n$}{$\m$} &\textit{$20$\quad\quad\,\,$25$\quad\quad\,\,$30$}& 	\textit{$20$\quad\quad\,$25$\quad\quad$30$}\\
\hline
$0.4$&$0.00$&$4$ 	&$\mathbf{0.064}$\quad0.064\quad0.064 &0.064\quad0.062\quad0.056\\
 &$0.30$ &  				&$\mathbf{0.172}$\quad0.161\quad0.139 &0.350\quad0.290\quad0.264\\
&$0.35$ &  					&$\mathbf{0.231}$\quad0.204\quad0.176 & 0.497\quad0.436\quad0.392\\
&$0.40$ &  					&$\mathbf{0.319}$\quad0.275\quad0.256 &0.659\quad0.605\quad0.546\\
&$0.45$ &  					&$\mathbf{0.425}$\quad0.389\quad0.334 &0.821\quad0.775\quad0.717\\
$0.7$&$0.00$ &			&$\mathbf{0.067}$\quad0.067\quad0.057 &0.054\quad0.059\quad0.049\\
 &$0.30$ &  				&$\mathbf{0.273}$\quad0.246\quad0.219 &0.664\quad0.584\quad0.542\\
&$0.35$ & 					&$\mathbf{0.393}$\quad0.363\quad0.321 						& 0.859\quad0.800\quad0.755\\
&$0.40$ &					&$\mathbf{0.571}$\quad0.515\quad0.465 						&0.970\quad0.947\quad0.908\\
&$0.45$ &					&$\mathbf{0.746}$\quad0.680\quad0.619 &0.997\quad0.990\quad0.982\\
\hline
$0.4$&$0.00$&$5$ 	&0.065\quad0.067\quad0.063 &0.059\quad$\mathbf{0.057}$\quad0.055\\
 &$0.30$ &  				&0.170\quad0.154\quad0.148 &0.335\quad$\mathbf{0.287}$\quad0.264\\
&$0.35$ &  					&0.227\quad0.202\quad0.179 &0.501\quad$\mathbf{0.428}$\quad0.388\\
&$0.40$ &  					&0.315\quad0.278\quad0.256 &0.667\quad$\mathbf{0.598}$\quad0.553\\
&$0.45$ &  					&0.429\quad0.386\quad0.355 &0.824\quad$\mathbf{0.775}$\quad0.715\\
$0.7$&$0.00$ &			&0.061\quad0.057\quad0.055 &0.049\quad0.041\quad$\mathbf{0.045}$\\
 &$0.30$ &  				&0.247\quad0.221\quad0.201 &0.647\quad$\mathbf{0.581}$\quad0.525\\
&$0.35$ & 					&0.393\quad0.353\quad0.318 &0.858\quad$\mathbf{0.797}$\quad0.727\\
&$0.40$ &					&0.571\quad0.495\quad0.438 &0.966\quad$\mathbf{0.940}$\quad0.905\\
&$0.45$ &					&0.725\quad0.658\quad0.598 &0.997\quad$\mathbf{0.990}$\quad0.983\\
\end{tabular}
\end{center}
\label{table:ex}
\end{table}

\section{An empirical illustration}\label{sec:emp}
To illustrate our testing procedure, we present an empirical  application concerning estimation of the effects of class size on students' performance on standardized tests. \cite{Angrist1999} studied the effects of class size on test scores of 4th and 5th grade students in Israel.
In this empirical illustration, we focus on 4th grade reading comprehension a feature that was also considered by \cite{Horowitz2011}.

In this empirical example we study the model
\begin{equation}
	\label{model:AL}
	Y_{sc} = \varphi(Z_{sc}, V_{sc})+D_{sc}\,\beta(V_{sc})
\end{equation}
where $Y_{sc}$ is the average reading comprehension test score of 4th grade students in class $c$ of school $s$, $Z_{sc}$ is the number of students in class $c$ of school $s$,  $D_{sc}$ is the fraction of disadvantaged students in class $c$ of school $s$ with unknown scalar function $\beta$,  $V_{sc}=U_{s}+\varepsilon_{sc}$ where $U_{s}$ is an unobserved school-specific random effect, and $\varepsilon_{sc}$ is an unobserved, independently over classes and schools distributed random variable.

The class size $Z_{sc}$ may be endogenous, for instance, due to the socioeconomic background of the students. To identify the causal effect of class size on scholar achievement \cite{Angrist1999} use \textit{Maimonides' rule} as instruments. According to this administrative rule, maximum class size is given by 40 pupils and will be split if the number of enrolled students exceeds this number. More precisely, assuming that cohorts are divided into classes of equal size, Maimonides’ rule is described by
\begin{equation*}
 W_{sc}=E_s / \lceil1 + (E_s-1) / 40\rceil 
\end{equation*}
where $E_s$ denotes enrollment in school $s$ and $\lceil x\rceil$ denotes the largest integer less or equal to $x$.
Note that \cite{Horowitz2011} could show that a linear relation between class size and scholar achievement as used by \cite{Angrist1999} is misspecified. 
To apply our tests, we consider a subsample where only one representative class per school is considered. By doing so, we avoid that rejection of a hypothesis may be caused by within class correlation. Moreover, only schools with at least two classes are considered which leads to a sample size of 707.

In the following, we want to test nonparametrically whether class size is endogenous at the $0.5$--quantile. 
The null hypothesis is that $\PP(Y_{sc}\leq \varphi(Z_{sc},q)+D_{sc}\,\beta(q)|Z_{sc})= q$ where $q=0.5$.
The value of our test statistic $S_n^{\textsl e}(0.5)=(2\m)^{-1/2}\big(4\, S_n^{\textsl e}(0.5)-\m\big)$ is given by $1.885$. For the choice of smoothing parameters $\k$ and $\m$ we applied the minimum-maximum principle as described in Section \ref{MC:sec:exog}. The resulting dimension parameters are $\k= 4$ and $\m=23$.\footnote{The value of the test for other  choices of $\k$ is  $2.254$ for $\k=3$ and $2.182$ for $\k=5$ where $\l=2\k$ and $\m$ is maximized over the range $k_n^2$ to $26$ (being the largest integer smaller than $\sqrt{707}$). }
  We thus reject the hypothesis of exogeneity at the $0.05$ nominal level. 
In particular, in model \eqref{model:AL} under conditions \textit{(a.1)}--\textit{(a.3)} we conclude that $Z_{sc}$ is not independent of $V_{sc}$.

We now test whether the model \eqref{model:AL} with conditions \textit{(a.1)}--\textit{(a.3)} is correctly specified. 
We construct our test statistic using B-splines as described in Section \ref{MC:sec:spec}. For the choice of smoothing parameters $\k$ and $\m$ we applied the minimum-maximum principle as described in Section \ref{MC:sec:exog}. As in the Monte Carlo section we choose $\l=2\k$. 
Our test statistic attains the value $ 1.4152$ and thus fails to reject the nonseparable model \eqref{model:AL} with conditions \textit{(a.1)}--\textit{(a.3)} at the $0.05$ nominal level. This value of the test statistic is obtained when $\k=4$ and $\m=26$. 
 For the fixed quantile $q=0.5$, we also performed a test of $\PP(Y_{sc}\leq \varphi(Z_{sc},q)+D_{sc}\,\beta(q)|W_{sc})= q$. In this case, our test statistic attains the value $0.981$ and again fails to reject the hypothesis.\footnote{This is not the case if $\k$ is chosen too small or too large. For instance if $\k=4$ or $\k=9$, respectively, then the value of the test statistic is $2.064$ or $3.420$ (as above maximized of $\m$ and $\l=2\k$). }

\begin{figure}[h]
	\centering
		\includegraphics[width=9.5cm]{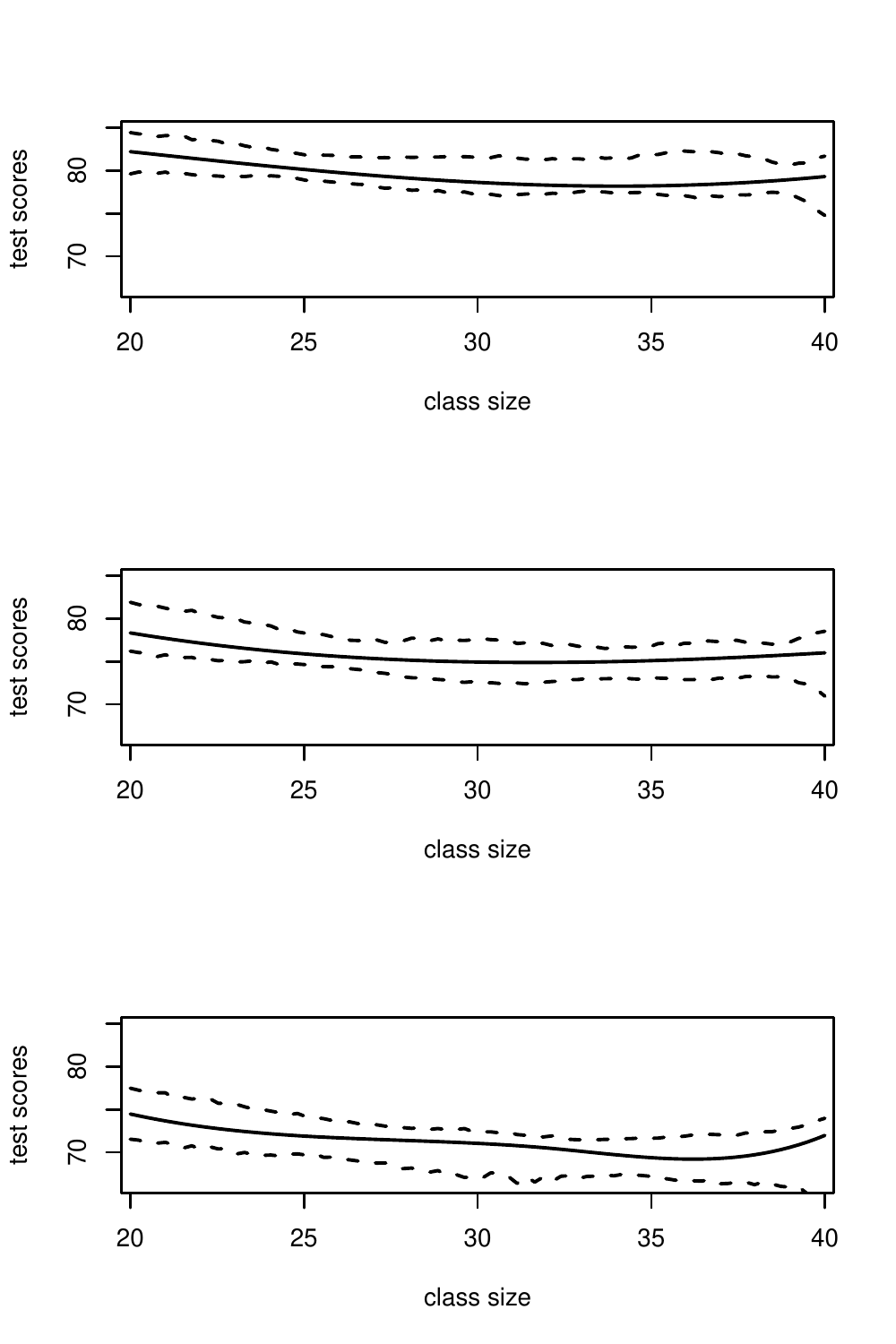}
	\caption{Estimated structural effects (solid lines) for $q\in\{0.75,\,0.5,\,0.25\}$ and $90$\% confidence intervals (dashed lines)}
	\label{quartiles6}
\end{figure}

 For the full sample, Figure \ref{quartiles6} depicts estimators of the structural effect $\solq$ for the quantiles  $q\in\{0.75,\,0.5,\,0.25\}$ where the number of disadvantaged students is restricted to be smaller than 15\% (which implies $n=688$). The solid lines are the estimators
and the dashed lines are the 90\% pointwise bootstrap confidence intervals using 1000 bootstrap iterations (we account for within school correlation by using schools as the bootstrap sampling units, see also \cite{Horowitz2011}). 
We can see that the confidence
intervals are tight enough to reject the hypothesis
that the quantile structural effects are overall upward sloping. In particular, we see that the effect of class size variation on test scores is more severe for lower performing classes.

\section{Conclusion}
In this paper, we develop a nonparametric specification test for the quantile regression model \eqref{model:NP}. The power of the test derives either from violations of regularity conditions imposed on the structural function, such as bounds or smoothness requirements, or a failure of monotonicity in the nonseparable unobservable variable. The test statistic is easy to implement and a natural extension of specification testing in a parametric framework. 
As the test builds on the sieve methodology, it allows to incorporate restrictions under the null hypothesis directly on the sieve space. As examples of tests of constraint hypotheses we consider in detail a test of exogeneity and a test of additivity of the structural function. We establish the large sample behavior of our test statistics and show that our tests work well in finite sample experiments. We also obtain reasonable results in an empirical illustration concerning the analysis of class size on students' performance. While we provide some heuristic guideline how to choose the sieve dimension in finite samples, an interesting future research area remains to provide asymptotic justification for it via adaptive testing.

%% file: proofs-asymt.tex
In the appendix, $f_\umn$ denotes an $\m$ dimensional vector with entries $f_j$ for $1\leq j\leq\m$. Moreover, $\|\cdot\|$ is the usual Euclidean norm.
For ease of notation, let $\bold X_i=(Y_i,Z_i,W_i)$ for $1\leq i\leq n$ with realizations $\bold x=(y,z,w)\in\cY\times\cZ\times \cW$. Let $\cH$ be a class of measurable functions with a measurable envelope function $H$. Then $N(\varepsilon, \cH , L_X^2)$ and $N_{[\,]}(\varepsilon,\cH , L_X^2)$, respectively, denote the covering and bracketing numbers for the set $\cH$. In addition, let $J_{[\,]}(1, \cH , L_X^2)$ denote a bracketing integral of $\cH$, that is,
\begin{equation*}
 J_{[\,]}(1, \cH , L_X^2)=\int_0^1\sqrt{1+\log N_{\small [\,]}(\varepsilon \,\|H\|_X,\cH , L_X^2)}d\varepsilon.
\end{equation*}
Throughout the proofs, we will use $C > 0$ to denote a generic finite constant that may be different in different uses. Further, for ease of notation we write $\int$ for $\int_0^1$, $\sum_i$ for $\sum_{i=1}^n$, and $\sum_{i'<i}$ for $\sum_{i=1}^n\sum_{i'=1}^{i-1}$.
For any $\phi,\psi\in L_W^2$, the inner product in $L_W^2$ is denoted by $\skalarV{\phi,\psi}_W=\Ex[\phi(W)\psi(W)]$ and let  $F_\m\phi=\sum_{j=1}^\m\skalarV{\phi,f_j}_Wf_j$.
In the following, we denote $\widehat Q_n=n^{-1}\sum_i f_\umn(W_i)f_\umn(W_i)^t$.
By Assumption \ref{ass:A1:q}, the eigenvalues of $\Ex[ f_\umn(W)f_\umn(W)^t]$ are bounded away from zero and hence, it may be assumed  that $\Ex[ f_\umn(W)f_\umn(W)^t]=I_\m$ where $I_\m$ denotes the $\m$ dimensional identity matrix (cf. \cite{Newey1997}, p. 161).

In the following result, we establish continuity of the mapping $q\mapsto\sol(\cdot,q)$ under the tangential cone condition and a mild assumption on the sieve approximation error for $\solq$. 
\begin{lem}
Let Assumption \ref{ass:A3:q} be satisfied. Assume for almost all $q\in(0,1)$ there exists a function $\sol_q$ with $\cT\solq=q$, let $T_q$ be compact,  and $\|\sol_{q}-\varPi_k\sol_{q}\|_Z=o( 1)$ as $k\to\infty$. Then the mapping $q\mapsto\sol(\cdot,q)$ is continuous.
\end{lem}
\begin{proof}
For some $q\in(0,1)$, since the linear operator $T_q$ is compact there exists singular value decomposition of it denoted by $\{ s_{qj}, e_j,f_j\}_{j\geq 1}$.
For any $\varepsilon>0$ and $k$ sufficiently large, let us define $\delta=(1-\eta)\,\varepsilon\, s_{qk}/3$.
We consider $q'\in(0,1)$ such that $|q-q'|<\delta$. Since $q,q'$ satisfy the quantile restriction we have $\|\cT \sol_{q}-\cT \sol_{q'}\|_W<\delta$. Let us further denote $r_k(q)=\|\varPi_k\sol_{q}-\sol_{q}\|_W$. We have $r_k(q)\leq \varepsilon/6$ by assumption for all $q$. 
By Assumption \ref{ass:A3:q} $(ii)$ and the triangular inequality it holds 
\begin{align*}
\|\cT \sol_{q}-&\cT \sol_{q'}\|_W\geq (1-\eta)\|T_{q} (\sol_{q}- \sol_{q'})\|_W \\
&=(1-\eta)\|T_{q}\varPi_k (\sol_{q}- \sol_{q'}) - T_{q}(\varPi_k\sol_{q}- \sol_{q})+T_{q}(\varPi_k\sol_{q'}- \sol_{q'})\|_W \\
&\geq (1-\eta)\Big(\|T_{q} \varPi_k(\sol_{q}- \sol_{q'})\|_W-\|T_{q}(\varPi_k\sol_{q}- \sol_{q})\|_W-\|T_{q}(\varPi_k\sol_{q'}- \sol_{q'})\|_W\Big)\\
&\geq (1-\eta)\,s_{qk}\Big(\|\varPi_k(\sol_{q}- \sol_{q'})\|_Z-r_k(q)-r_k(q')\Big)\\
&\geq (1-\eta)\,s_{qk}\Big(\|\sol_{q}- \sol_{q'}\|_Z-2r_k(q)-2r_k(q')\Big),
\end{align*}
using that $( s_{qj})_{j\geq 1}$ is a nonincreasing sequence. 
  This implies 
  \begin{align*}
  \|\sol_{q}- \sol_{q'}\|_Z&\leq (1-\eta)^{-1}\,s_{qk}^{-1}\,\delta + 2r_k(q)+2r_k(q')\\
  &\leq (1-\eta)^{-1}\,s_{qk}^{-1}\,\delta +2\, \varepsilon/3\\
  &\leq \varepsilon,
  \end{align*}
  which proves the result. 
\end{proof}
  
  \begin{proof}[\textsc{Proof of Proposition \ref{prop:overident}.}]
  Let $\|T_q\|_{o,2}$ be the operator norm of the Fr\'echet derivative  $T_q$ given by $\|T_q\|_{o,2}=\sup_{\{\phi\in L_Z^2:\, \|\phi\|_Z\leq 1\}} \|T_q\phi\|_W$.  From Assumption \ref{ass:A0} $(v)$ we infer that the operator $T_q$ is bounded since
  \begin{align*}
  \|T_q\|_{o,2}^2&\leq \sup_{\{\phi\in L_Z^2:\, \|\phi\|_Z\leq 1\}} \Ex|\Ex\big[ p_{Y|Z,W}\big(\sol(Z,q),Z,W\big)\phi(Z)\big|W\big]|^2\\
 & \leq C \sup_{\{\phi\in L_Z^2:\, \|\phi\|_Z\leq 1\}} \Ex|\phi(Z)|^2\\
  & \leq C.
  \end{align*}
  Since $\|\cdot\|_Z \leq  \|\cdot\|_{\alpha,p}$ for any integer $p>0$ (see e.g. Lemma A.2 of \cite{Santos12}) we have $\sup_{\phi\in\mathcal B} \|\phi\|_Z \leq \rho$ by the definition of $\mathcal B$. Consequently, for any $\phi\in\mathcal B$ we obtain 
\begin{align*}
\|T_q\phi\|_W\leq \|T_q\|_{o,2}\|\phi\|_Z\leq \rho\,\|T_q\|_{o,2}.
\end{align*}  
We conclude that the range  $\mathcal R_q(\mathcal B)$ is uniformly bounded by the constant $\rho\,\|T_q\|_{o,2}$ and hence,  $\mathcal R_q(\mathcal B)$ is a strict subset of $L_W^2$, which completes the proof. 
\end{proof}
  
\begin{proof}[\textsc{Proof of Theorem \ref{prop:main}.}]
Since we have $\|\widehat Q_n-I_\m\|^2=o_p(m_n^2/n)$ it is sufficient to prove that
$3\sqrt{5/\m}\big(\sum_{j=1}^\m \int |n^{-1/2}\sum_i  (\1{\{Y_i\leq\hsol_{qn}(Z_i)\}}-q)f_j(W_i)|^2dq-\m/6\big)\stackrel{d}{\rightarrow}\cN(0,1)$.
The proof is based on the decomposition
 \begin{multline}\label{proof:ex:dec}
 \sum_{j=1}^\m\int\big|n^{-1}\sum_i (\1{\{Y_i\leq\hsol_{qn}(Z_i)\}}-q)f_j(W_i)\big|^2dq\\
 =\sum_{j=1}^\m\int\big|n^{-1}\sum_i (\1{\{Y_i\leq\sol(Z_i,q)\}}-q)f_j(W_i)\big|^2dq\\
-\frac{2}{n^2}\sum_{j=1}^\m\int \Big(\sum_i (\1{\{Y_i\leq\sol(Z_i,q)\}}-q)f_j(W_i)\Big)\\\hfill
\times\Big(\sum_i \big(\1\{Y_i\leq \hsolq(Z_i)\}-\1\{Y_i\leq\sol(Z_i,q)\}\big)f_j(W_i)\Big)dq\\
+\sum_{j=1}^\m\int\big|n^{-1}\sum_i \big(\1\{Y_i\leq \hsolq(Z_i)\}-\1\{Y_i\leq\sol(Z_i,q)\}\big)f_j(W_i)\big|^2dq
=I_n-2II_n+III_n.
 \end{multline}
Consider $I_n$. We calculate further
\begin{multline*}
m_n^{-1/2}\big(nI_n-\m/6\big)
=\frac{1}{\sqrt\m n}\sum_i \sum_{j=1}^\m\Big(\int|(\1{\{Y_i\leq\sol(Z_i,q)\}}-q)f_j(W_i)|^2dq-1/6\Big)\\
+\frac{1}{\sqrt\m n}\sum_{i\neq i'}\sum_{j=1}^\m \int\big(\1{\{Y_i\leq\sol(Z_i,q)\}}-q\big)\big(\1{\{Y_{i'}\leq\sol(Z_{i'},q)\}}-q\big)f_j(W_i)f_j(W_{i'})dq
\end{multline*}
where the first summand tends in probability to zero as $n\to\infty$. Indeed,we have
\begin{equation*}
 \Ex \int|(\1\{Y\leq \sol(Z,q)\}-q)f_j(W)|^2dq=\Ex[f_j^2(W)]\int q(1-q)dq=1/6
\end{equation*}
for all $j\geq 1$ and hence, 
\begin{multline*}
 \Ex\Big|\frac{1}{\sqrt\m n}\sum_i \sum_{j=1}^\m\Big(\int|(\1{\{Y_i\leq\sol(Z_i,q)\}}-q)f_j(W_i)|^2dq-1/6\Big)\Big|^2\\\hfill
\leq \frac{1}{\m n}\int\Ex\Big|\sum_{j=1}^\m|(\1{\{Y\leq \sol(Z,q)\}}-q)f_j(W)|^2-\Ex|(\1\{Y\leq \sol(Z,q)\}-q)f_j(W)|^2\Big|^2dq\\\hfill
\leq \frac{1}{\m n}\sup_{w\in\cW}\|f_\umn(w)\|^4\int\Ex|\1{\{Y\leq \sol(Z,q)\}}-q|^4dq\\
\leq O( m_n/n)=o(1)
\end{multline*}
by using  $\sup_{w\in\cW}\|f_\umn(w)\|^2\leqslant C\m$.
Therefore, to establish $3\sqrt{5/\m}(nI_n-\m/6)\stackrel{d}{\rightarrow}\cN(0,1)$ it is sufficient to show
\begin{equation*}\label{norm:result:ind}
\frac{3\sqrt 5}{\sqrt\m n}\sum_{i\neq i'}\sum_{j=1}^\m \int\big(\1{\{Y_i\leq \sol(Z_i,q)\}}-q\big)\big(\1{\{Y_{i'}\leq \sol(Z_{i'},q)\}}-q\big)f_j(W_i)f_j(W_{i'})dq\stackrel{d}{\rightarrow} \cN(0,1).
\end{equation*}
This follows from Lemma \ref{normal:lem:1}.
Consider $III_n$. 
Let us denote $\cB_{n}:=\{\phi\in\cB^{(0,1)}:\,\interleave\phi-\sol\interleave_{Z,p}^2\leq m_n^{-(1+c)/\kappa}\}$ for some constant $c>0$ and $\cB_{qn}:=\{\phi_q:\,\phi\in\cB_n\}\subset\cB$.
Further, we denote for $1\leq j\leq \m$ and $1\leq i\leq n$ 
\begin{equation*}
 h_{qj}(\bold X_i,\phi_q)=\big(\1\{Y_i\leq \phi(Z_i,q)\}-\1{\{Y_i\leq\sol(Z_i,q)\}}\big)f_j(W_i)
\end{equation*}
and the classes $\cH_{qjn}=\{h_{qj}(\cdot,\phi_q):\,\phi_q\in\cB_{qn}\}$ and $\cH_{qj}=\{h_{qj}(\cdot,\phi_q):\,\phi_q\in\cB\}$.
We observe
\begin{multline*}
III_n=\sum_{j=1}^\m \int\big|n^{-1}\sum_i  h_{qj}(\bold X_i,\hsolq)\big|^2dq\\
\leq 2\eta_p \interleave\cT\hsol_{\cdot n}-\cT\sol\interleave_W^2
+2\sum_{j=1}^\m \int\big|n^{-1}\sum_i  h_{qj}(\bold X_i,\hsolq)-\skalarV{\cT\hsolq-\cT\solq,f_j}_W\big|^2dq.
\end{multline*}
%
From  \eqref{lem:eq:2} in Lemma \ref{lem:1} together with condition $n\tau_n=o(\sqrt{\m})$ we deduce $n\interleave\cT\hsol_{\cdot n}-\cT\sol\interleave_W^2=o_p(\sqrt\m)$.
Further, we observe for every $\phi_q\in\cB_{qn}$ that
\begin{equation*}
 \big|h_{qj}(\bold X_i,\phi_q)\big|^2\leq \max_{\phi_q\in\cB_{qn}}\big|\big(\1\{Y_i\leq \phi(Z_i,q)\}-\1{\{Y_i\leq\sol(Z_i,q)\}}\big)f_j(W_i)\big|^2=:H_{qj}^2(\bold X_i)
\end{equation*}
and hence, $H_{qj}$ is an envelope function of the class $\cH_{qjn}$ and due to Assumption \ref{ass:A5:q} we have $\Ex[\int H_{qj}^2(\bold X)dq]\leq Cm_n^{-(1+c)}$.
Moreover, \eqref{lem:eq:3} in Lemma \ref{lem:1} together with condition \eqref{prop:chen:est:cond} implies $\interleave\hsol_{\cdot n}-\sol\interleave_{Z,p}^2=o_p\big(m_n^{-(1+c)/\kappa}\big)$ and thereby
\begin{multline*}
 \PP\Big(\sum_{j=1}^\m\int\big|n^{-1}\sum_i  h_{qj}(\bold X_i,\hsolq)-\skalarV{\cT\hsolq-\cT\solq,f_j}_W\big|^2dq>\varepsilon\Big)\\\hfill
\leq \sum_{j=1}^\m \varepsilon^{-1}\Ex\sup_{\phi\in\cB_{n}}\int\Big|n^{-1/2}\sum_i h_{qj}(\bold X_i,\phi_q)-\Ex h_{qj}(\bold X,\phi_q)\Big|^2dq+o(1)\\\hfill
\leq \sum_{j=1}^\m \varepsilon^{-1}\int\Ex\max_{\phi_q\in\cB_{qn}}\Big|n^{-1/2}\sum_i h_{qj}(\bold X_i,\phi_q)-\Ex h_{qj}(\bold X,\phi_q)\Big|^2dq+o(1)\\\hfill
\leq \sum_{j=1}^\m \varepsilon^{-1}\int\Big(\Ex\max_{\phi_q\in\cB_{qn}}\Big|n^{-1/2}\sum_i h_{qj}(\bold X_i,\phi_q)-\Ex h_{qj}(\bold X,\phi_q)\Big|+\big(\Ex|H_{qj}(\bold X)|^2\big)^{1/2}\Big)^2dq+o(1)
\end{multline*}
where the last inequality is due to Theorem 2.14.5 of \cite{Vaart2000}. 
We further conclude by applying the last display of Theorem 2.14.2 of \cite{Vaart2000}
\begin{equation*}
\Ex\max_{\phi_q\in\cB_{qn}}\Big|n^{-1/2}\sum_i h_{qj}(\bold X_i,\phi_q)-\Ex h_{qj}(\bold X,\phi_q)\Big|\leq C J_{[\,]}(1,\cH_{qjn},L^2_{\bold X})\big(\Ex|H_{qj}(\bold X)|^2\big)^{1/2}
\end{equation*}
for all  $0<q<1$.
Now since $\max_{1\leq j\leq \m}\Ex\int|H_{qj}(\bold X)|^2dq\leq Cm_n^{-(1+c)}$ for $n$ sufficiently large it is sufficient to show that $\max_{1\leq j\leq \m}J_{[\,]}(1,\cH_{qjn},L^2_{\bold X})<C$ for all  $0<q<1$. 
From Lemma 4.2 (i) of \cite{Chen07} we deduce
\begin{align*}
 N_{[\,]}(\varepsilon\big(\Ex|H_{qj}(\bold X)|^2\big)^{1/2},\cH_{qjn},L^2_{\bold X})
&\leq N_{[\,]}\Big(\varepsilon,\big(\Ex|H_{qj}(\bold X)|^2\big)^{-1/2}\,\cH_{qjn},L^2_{\bold X}\Big)\\
&\leq N_{[\,]}\Big(\varepsilon,\cH_{qj},L^2_{\bold X}\Big)\\
&\leq N\Big(\Big(\frac{\varepsilon}{2 C}\Big)^{2/\kappa},\cB,\|\cdot\|_{Z,p}\Big)\\
&\leq N\Big(\Big(\frac{\varepsilon}{2 C}\Big)^{2/\kappa},\cB,\|\cdot\|_\infty\Big).
\end{align*}
Employing condition $\alpha_0>{d_z}/p$ and Theorem 6.2 Part II of \cite{2003Adams} yields that  $W^{\alpha,p}$ is compactly embedded in $W^{\alpha,\infty}$. Thereby, $\cB\subset  W^{\alpha,p}$ is totally bounded in $W^{\alpha,\infty}$ which implies $\|\phi\|_{\alpha,\infty}\leq C$ for all $\phi\in\cB$. Let $W^{\alpha,\infty}_C:=\{W^{\alpha,\infty}:\,\|\phi_q\|_{\alpha,\infty}\leq C\}$. Now Theorem 2.7.1 of \cite{Vaart2000} gives
\begin{equation*}
 \log N\big(\varepsilon^{2/\kappa},\cB,\|\cdot\|_\infty\big)\leq \log N\big(\varepsilon^{2/\kappa},W^{\alpha,\infty}_C,\|\cdot\|_\infty\big)
\leq C\varepsilon^{-2{d_z}/(\alpha\kappa)}
\end{equation*}
where $C$ depends on the diameter of $\cZ$.
%
%
Now due to Assumption \ref{ass:A0} (i) it is straightforward to see that $\max_{1\leq j\leq \m}J_{[\,]}(1,\cH_{qjn},L^2_{\bold X})<C$ and hence, $nIII_n=o_p(\sqrt\m)$.

Consider $II_n$. We observe
\begin{multline*}
 nII_n= \sum_{j=1}^\m \int\Big(\sum_i (\1{\{Y_i\leq\sol(Z_i,q)\}}-q)f_j(W_i)\Big)\Big(n^{-1}\sum_i h_{qj}(\bold X_i,\hsolq)\Big)dq\\\hfill
= \sum_{j=1}^\m \int\Big(\sum_i (\1{\{Y_i\leq\sol(Z_i,q)\}}-q)f_j(W_i)\Big)\Big(n^{-1}\sum_i h_{qj}(\bold X_i,\hsolq)-\skalarV{\cT\hsolq-\cT\solq,f_j}_W\Big)dq\\\hfill
+\sum_{j=1}^\m \int\Big(\sum_i (\1{\{Y_i\leq\sol(Z_i,q)\}}-q)f_j(W_i)\Big)\skalarV{\cT\hsolq-\cT\solq,f_j}_Wdq\\
=C_{n1}+C_{n2}.
\end{multline*}
The Cauchy Schwarz inequality implies for all $\varepsilon>0$
\begin{multline*}
 \PP(|C_{n1}|>\varepsilon\sqrt\m)
\leq (\varepsilon\sqrt\m)^{-1}\Big(\int q(1-q)dq\Big)^{1/2}\\\hfill
\times\sum_{j=1}^\m \Big(\int\Ex\max_{\phi_q\in\cB_{qn}}\big|n^{-1/2}\sum_i h_{qj}(\bold X_i,\phi_q)-\Ex h_{qj}(\bold X,\phi_q)\big|^2dq\Big)^{1/2}+o(1)\\
=o(1)
\end{multline*}
where the last equality follows similarly to the proof of $nIII_n=o_p(\sqrt\m)$. 
Consider $C_{n2}$. 
Let us introduce the function for $1\leq j\leq \m$ and $1\leq i\leq n$ 
\begin{equation*}
 t_{qn}(\bold X_i,\phi_q):=\big(\1{\{Y_i\leq\sol(Z_i,q)\}}-q\big)\big(F_\m\cT\phi_q-F_\m\cT\solq\big)(W_i)
\end{equation*}
and the sets $\cD_{n}:=\set{\phi\in\cB^{(0,1)}:n\interleave\cT\phi-\cT\sol\interleave_W^2\leq \sqrt{\m}}$, $\cD_{qn}:=\set{\phi_q:\phi\in\cD_{n}}\subset\cB$, $\cG_q:=\{t_{qn}:\,\phi\in\cB\}$, and $\cG_{qn}:=\{t_{qn}:\,\phi\in\cD_{qn}\}$. We calculate
\begin{equation*}
 \PP\big(|C_{n2}|>\varepsilon\sqrt\m\big)\leq \sqrt n(\varepsilon\sqrt\m)^{-1}\Ex\int\max_{\phi_q\in\cD_{qn}}\Big|\frac{1}{\sqrt n}\sum_i t_{qn}(\bold X_i,\phi_q)\Big|dq+o(1).
\end{equation*}
Since $p_W$ is uniformly bounded away from zero, $n\interleave\cT\phi-\cT\sol\interleave_W^2\leq \sqrt{\m}$,  and $\|F_\m(\cT\phi_q-\cT\solq)\|_W\leq C\|\cT\phi_q-\cT\solq\|_W$ for all $\phi\in\cD_{n}$ we have  $|F_\m(\cT\phi_q-\cT\solq)(w)|\leq C\,m_n^{1/4}n^{-1/2}$ for  almost all $0<q<1$ and $p_W$--almost all $w$.
Consequently, $t_{qn}(\bold x,\phi_q) \leq C\,m_n^{1/4}n^{-1/2}$ $p_W$--almost surely.
We conclude by again applying the last display of Theorem 2.14.2 of \cite{Vaart2000}
\begin{equation*}
 \Ex\max_{\phi_q\in\cD_{qn}}\Big|\frac{1}{\sqrt n}\sum_i t_{qn}(\bold X_i,\phi_q)\Big|
\leq C J_{[\,]}(1,\cG_{qn},L_{\bold X}^2)\,m_n^{1/4}n^{-1/2}.
\end{equation*}
As above it can be seen that $J_{[\,]}(1,\cG_{qn},L^2_{\bold X})<C$ for all $0<q<1$. 
%
Indeed, from Assumption \ref{ass:A3:q} $(ii)$ we conclude $\|\cT\phi-\cT\solq\|_W\leq (1+\eta)\|T_q(\phi-\solq)\|_W$ and further, Assumption \ref{ass:A0} $(v)$ yields $\|F_\m(\cT\phi-\cT\solq)\|_W\leq C(1+\eta)\eta_p\|\phi-\solq\|_Z$. 
Hence,  the mapping $\phi\mapsto F_\m\cT\phi$ is Lipschitz continuous at $\solq$ and we may apply Theorem 2.7.11 of \cite{Vaart2000} which yields
\begin{align*}
 N_{[\,]}(\varepsilon\big(n^{-1}\sqrt\m\big)^{1/2},\cG_{n},L^2_{\bold X})&
\leq N_{[\,]}(\varepsilon,\cG_q,L^2_{\bold X})\\
&\leq N_{[\,]}\big(\varepsilon,\{F_\m\cT\phi-F_\m\cT\solq:\,\phi\in\cB\},L^2_W\big)\\
&\leq N\Big(\frac{\varepsilon}{2C},\cB,\|\cdot\|_\infty\Big).
\end{align*}
Thereby, $C_{n2}=o_p(\sqrt\m)$, which completes the proof. 
\end{proof}

In the following we make use of  the notation $g_{qj}(\bold X_i,\phi):=(\1\{Y_i\leq \phi(Z_i)\}-q)f_j(W_i)$, $1\leq j\leq \m$, $1\leq i\leq n$, for any $\phi\in\cB$.
\begin{proof}[\textsc{Proof of Proposition \ref{prop:cons:q}.}] 
For the proof it is sufficient to show $n^{-1}S_n\geq\int\|\cT\solq-q\|_W^2dq/2+o_p(1)$.
Since $\int\|n^{-1}\sum_i (\1{\{Y_i\leq\hsolq(Z_i)\}}-\1{\{Y_i\leq\solq(Z_i)\}})f_\umn(W_i)\|^2dq=o_p(1)$  (cf. proof of Theorem \ref{prop:main} together with Lemma \ref{lem:1}) we obtain
\begin{align*}
 \int\big\|n^{-1}\sum_i (\1{\{Y_i\leq\solq(Z_i)\}}-q)&f_\umn(W_i)\big\|^2dq\\
 &=\int\big\|\Ex[((\cT\solq)(W)-q)f_\umn(W)]\big\|^2dq+o_p(1)\\
&\geq \int\|\cT\solq-q\|_W^2dq/2+o_p(1),
\end{align*}
which proves the result.
\end{proof}

\begin{proof}[\textsc{Proof of Proposition  \ref{coro:norm:q}.}]
Since $\solq=\argmin_{\phi\in\cB}\|\cT\phi-q\|_W$ we obtain as in the proof of Theorem \ref{prop:main} by employing the results of Lemma \ref{lem:1} that
\begin{equation*}
 S_n=\sum_{j=1}^\m\int\Big|n^{-1/2}\sum_i g_{qj}(\bold X_i,\solq)\Big|^2dq+o_p(\sqrt\m).
\end{equation*}
Further, we calculate
\begin{multline*}
\sum_{j=1}^\m\int\Big|n^{-1/2}\sum_i g_{qj}(\bold X_i,\solq)\Big|^2dq
=\sum_{j=1}^\m\int\Big|n^{-1/2}\sum_i \big(g_{qj}(\bold X_i,\solq)-\Ex g_{qj}(\bold X_i,\phi)\big)\Big|^2dq\\\hfill
+2\sum_{j=1}^\m\int\Big(n^{-1/2}\sum_i \big(g_{qj}(\bold X_i,\solq)-\Ex g_{qj}(\bold X,\solq)\big)\Big)\sqrt n\Ex g_{qj}(\bold X,\solq)dq\\\hfill
+n\sum_{j=1}^\m\int\Big|\Ex g_{qj}(\bold X,\solq)\Big|^2dq\\
=I_n+2II_n+III_n.
\end{multline*}
We have $ 3\sqrt{5/\m}\big(I_n-\m/6\big)\stackrel{d}{\rightarrow} \cN(0,1)$.
Further, since $\Ex[(\1\{Y\leq\sol(Z,q)\}-q)^2|W]\leq 1$ we obtain
\begin{multline}\label{ineq:II}
 \Ex|II_n|^2\leq n\Ex\int\big|(\1\{Y\leq\sol(Z,q)\}-q)\sum_{j=1}^\m\Ex g_{qj}(\bold X,\solq)\big|^2dq\\\hfill
 \leq n\int\Big|\sum_{j=1}^\m\Ex g_{qj}(\bold X,\solq)\Big|^2dq
 \leq n\int\|\cT\solq-q\|_W^2dq
\end{multline}
and hence $II_n=O_p( (n\int\|\cT\solq-q\|_W^2dq)^{1/2})$.
Moreover, since $n\int\|\cT\sol_{qn}-\cT\solq\|_W^2dq=o(\sqrt\m)$ and by employing relation  \eqref{loc:alt:ind:q} it is easily seen that
\begin{equation*}
 \frac{3\sqrt5}{\sqrt\m}III_n=\frac{3\sqrt5 \, n}{\sqrt\m}\int\|\cT\sol_{qn}-q-\delta_n\xi_q\|_W^2dq+\sum_{j=1}^\infty\int\Ex[\xi_q(W)f_j(W)]^2dq+o(1),
\end{equation*}
which proves the result.
\end{proof}

\begin{proof}[\textsc{Proof of Corollary \ref{prop:main:B:power}.}]
For the proof it is sufficient to show $n^{-1}S_n^*\geq\int\|\cT\solq-q\|_W^2dq/2+o_{p^*}(1)$ with probability approaching one. \cite{chen2015} show that the bootstrap version of the sieve estimator $\widehat \sol_{qn}^*$ converges at the same rate as $\hsolq$. In light of the proof of Proposition \ref{prop:cons:q}, it is sufficient to show
\begin{align*}
 \int\big\|n^{-1}\sum_i \varepsilon_i(\1{\{Y_i\leq\solq(Z_i)\}}-q)&f_\umn(W_i)\big\|^2dq\\
 &=\int\big\|\Ex[((\cT\solq)(W)-q)f_\umn(W)]\big\|^2dq+o_p(1)\\
&\geq \int\|\cT\solq-q\|_W^2dq/2+o_p(1),
\end{align*}
using that $\varepsilon$ is independent of $W$ and $\Ex[\varepsilon]=1$ as well as $\Var(\varepsilon)<\infty$,
which proves the result.
\end{proof}

\begin{proof}[\textsc{Proof of Corollary  \ref{coro:int}.}]
In light of the proof of Theorem \ref{prop:main} it is sufficient to prove $n\|\cT\hsol_{qn}^{\textsl e}-\cT\solq\|_W^2=o_p(\sqrt\m)$. Due to Assumption \ref{ass:ex} $(ii) $ we obtain as in the proof of Theorem 6 of \cite{chen2012} that
\begin{equation*}
 \|\cT\hsol_{qn}^{\textsl e}-\cT\solq-T_q(\hsol_{qn}^{\textsl e}-\solq)\|_W\leq C\|\hsol_{qn}^{\textsl e}-\solq\|_Z^2
\end{equation*}
and consequently,
\begin{equation*}
 \|\cT\hsol_{qn}^{\textsl e}-\cT\solq\|_W\leq C\big(\|T_q(\hsol_{qn}^{\textsl e}-\solq)\|_W+\|\hsol_{qn}^{\textsl e}-\solq\|_Z^2\big).
\end{equation*}
Moreover, by applying $\sup_{y}p_{Y|Z,W}(y,Z,W)\leq C$ and Jensen's inequality we have
\begin{align*}
 \|T_q(\hsol_{qn}^{\textsl e}-\solq)\|_W^2&=\int_\cW|\int_\cZ p_{Y|Z,W}(\sol(z,q),z,w)(\hsol_{qn}^{\textsl e}-\solq)(z)p_{Z|W}(z,w)dz|^2p_W(w)dw\\
 &\leq C\|\hsol_{qn}^{\textsl e}-\solq\|_Z^2\\
 &=O_p(R_n^{\textsl e})\\
 &=o_p(\sqrt\m/n),
\end{align*}
by employing the rate conditions \eqref{cond:ex} and Assumption \ref{ass:ex} $(iii)$.
\end{proof}

\begin{proof}[\textsc{Proof of Proposition  \ref{loc:power:ex}.}]
Due to the rate restriction \eqref{cond:ex} we may follow the proof of Theorem \ref{prop:main} and Corollary  \ref{coro:int} and hence obtain
\begin{equation*}
 S_n^{\textsl e}(q)=\sum_{j=1}^\m\Big|n^{-1/2}\sum_i (\1\{Y_i\leq\sol_q^{\textsl e}(Z_i)\}-q)f_j(W_i)\Big|^2+o_p(\sqrt\m).
\end{equation*}
Thus, by following line by line the Proposition \ref{coro:norm:q}, we obtain the result.
\end{proof}

%% file: proofs-tech.tex
\subsection{Technical assertions.}
We can not apply the consistency and rate of convergence results of \cite{chen08} when the null hypothesis $H_0$ fails. The following Lemma extends their results to possibly misspecified instrumental quantile regression. Recall that under misspecification $\solq=\argmin_{\phi\in\cB}\|\cT\phi-q\|_W$ does not satisfy $\cT\solq=q$.
\begin{lem}\label{lem:1}
 Let Assumptions \ref{ass:A1:q}--\ref{ass:A0} hold true. Then 
 \begin{gather}
   \interleave\hsol_{\cdot n}-\sol\interleave_{Z,p}^2=o_p(1)\label{lem:eq:1},\\
  \interleave\cT\hsol_{\cdot n}-\cT\sol\interleave_W^2=O_p\Big(\omega_n+\int\|\cT\solq-q\|_W^2dq\Big),\label{lem:eq:2}\\
    \interleave\hsol_{\cdot n}-\sol\interleave_{Z,p}^2=O_p\Big(\interleave\varPi_\k\sol-\sol\interleave_{Z,p}^2+\tau_\k\big(\omega_n+\int\|\cT\solq-q\|_W^2dq\big)\Big).\label{lem:eq:3}
 \end{gather}
\end{lem}
\begin{proof}
 Proof of \eqref{lem:eq:1}. We define $\cR_n:=\max\big(n^{-1}\l,\max_{\phi\in\cB_\k}\sum_{j>\l}\Ex[(\cT\phi(W)-q)f_j(W)]^2\big)$.
From the proof of Proposition \ref{prop:cons:q} we have that
\begin{equation}\label{lem:eq}
\sum_{j=1}^\l \Ex\max_{\phi\in\cB_\k}\Big|n^{-1}\sum_i\1\{Y_i\leq \phi(Z_i)\}f_j(W_i)-\Ex[\1\{Y\leq \phi(Z)\}f_j(W)]\Big|^2=O(n^{-1}\l).
\end{equation}
Consequently, we observe
 \begin{equation*}
  \int\big\|n^{-1}\sum_i (\1\{Y_i\leq \varPi_\k\solq(Z_i)\}-q)f_\uln(W_i)\big\|^2dq
  \leq 2\int\|\cT\varPi_\k\solq-q\|_W^2dq +O_p(\cR_n).
 \end{equation*} 
Further, using the elementary inequality $(a-b)^2\geq a^2/2-b^2$ and again applying relation \eqref{lem:eq} gives
\begin{multline*}
  \int\big\|n^{-1}\sum_i (\1\{Y_i\leq \phi_q(Z_i)\}-q)f_\uln(W_i)\big\|^2dq
  \geq \int\|F_\l(\cT\phi_q-q)\|_W^2dq/2\\\hfill-\sum_{j=1}^\l \max_{\phi\in\cB_\k}\Big|n^{-1}\sum_i\1\{Y_i\leq \phi(Z_i)\}f_j(W_i)-\Ex\1\{Y\leq \phi(Z)\}f_j(W)\Big|^2\\
  \geq C\int\|\cT\phi_q-q\|_W^2dq- O_p(\cR_n).
 \end{multline*}
 Let us denote $\cA_\k=\{\phi\in\cB_\k^{(0,1)}:\,\interleave\phi-\sol\interleave_{Z,p}^2\geq\varepsilon\}$ for some $\varepsilon>0$. Since $\cT$ is continuous and $\solq=\argmin_{\phi\in\cB}\|\cT\phi-q\|_W$ is unique  we have that $\min_{\phi\in\cA_\k}\int\|\cT\phi_q-q\|_W^2dq$ is strictly positive for all $n\geq 1$. 
 Therefore, we obtain
 \begin{multline*}
  \PP\Big(\interleave\hsol_{\cdot n}-\sol\interleave_{Z,p}^2\geq \varepsilon\Big)\\
  \leq \PP\Big(\min_{\phi\in\cA_\k} \int\big\|\sum_i (\1\{Y_i\leq \phi(Z_i,q)\}-q)f_\uln(W_i)\big\|^2dq\\\hfill
  \leq \int\big\|\sum_i (\1\{Y_i\leq \varPi_\k\sol(Z_i,q)\}-q)f_\uln(W_i)\big\|^2dq\Big)\\\hfill
  \leq \PP\Big(\min_{\phi\in\cA_\k} \int\|\cT\phi_q-q\|_W^2dq
  \leq \int\|\cT\varPi_\k\solq-q\|_W^2dq+O_p(\cR_n)\Big)
  =o(1)
 \end{multline*}
 since $\int\|\cT\varPi_\k\solq-q\|_W^2dq=\int\|\cT\solq-q\|_W^2dq+o(1)$, $\cR_n=o(1)$, and making use of $\min_{\phi\in\cA_\k} \int\|\cT\phi_q-q\|_W^2dq>\int\|\cT\solq-q\|_W^2dq+o(1)$.
 Proof of \eqref{lem:eq:2}. 
 For some $\varepsilon>0$ let us denote $\cD_{\k}=\{\phi\in\cB_\k^{(0,1)}:\,\interleave\cT\phi-\cT\sol\interleave_W^2\geq\varepsilon\omega_n\}$.
 Therefore, we obtain as above
 \begin{multline*}
  \PP\Big(\interleave\cT\hsol_{\cdot n}-\cT\sol\interleave_W^2\geq \varepsilon \omega_n\Big)\\\hfill
  \leq \PP\Big(\min_{\phi\in\cD_{\k}} \int\|\cT\phi_q-q\|_W^2dq\leq \int\|\cT\varPi_\k\solq-q\|_W^2dq+O_p(\cR_n)\Big).
 \end{multline*}
 Further, it holds $\int\|\cT\varPi_\k\solq-q\|_W^2dq\leq 2\interleave\cT\varPi_\k\sol-\cT\sol\interleave_W+2\int\|\cT\solq-q\|_W^2dq$.
   We thus obtain
 \begin{multline*}
  \PP\Big(\interleave\cT\hsol_{\cdot n}-\cT\sol\interleave_W^2\geq \varepsilon \,\omega_n\Big)\\
  \leq \PP\Big(\min_{\phi\in\cD_{\k}}\int\|\cT\phi_q-q\|_W^2dq\leq 2\interleave\cT\varPi_\k\sol-\cT\sol\interleave_W^2+2\int\|\cT\solq-q\|_W^2dq+O_p(\cR_n)\Big).
\end{multline*}
For all $\phi\in\cD_{\k}$ and $0<q<1$ we have
 \begin{equation*}
  \|\cT\phi_q-q\|_W^2\geq \|\cT\solq-q\|_W^2\geq \|\cT\phi_q-\cT\solq\|_W^2/2-\|\cT\phi_q-q\|_W^2
 \end{equation*}
and hence, $\|\cT\phi_q-q\|_W^2\geq \|\cT\phi_q-\cT\solq\|_W^2/4$. Thereby, we obtain
 \begin{multline*}
  \PP\Big(\interleave\cT\hsol_{\cdot n}-\cT\sol\interleave_W^2\geq \varepsilon\, \omega_n\Big)\\\hfill
  \leq \PP\Big(\frac{1}{4}\min_{\phi\in\cD_{\k}}\interleave\cT\phi-\cT\sol\interleave_W^2\leq 2\interleave\cT\varPi_\k\sol-\cT\sol\interleave_W^2+2\int\|\cT\solq-q\|_W^2dq + O_p(\cR_n)\Big)\\
  \leq \PP\Big(\frac{\varepsilon}{4}\omega_n \leq 2\eta\int\|T_q(\varPi_\k\solq-\solq)\|_W^2dq+2\int\|\cT\solq-q\|_W^2dq + O_p(\cR_n)\Big)
\end{multline*}
which goes to zero for all $n\geq 1$ as $\varepsilon\to\infty$.
Proof of \eqref{lem:eq:3}.
 Note that $\|T_q(\phi-\solq)\|_W\leq (1-\eta)^{-1}\|\cT\phi-\cT\solq\|_W$ for all $\phi$ in a sufficiently small neighborhood around $\solq$. Thereby, due to \eqref{lem:eq:1} we obtain
\begin{equation*}
 \interleave\hsol_{\cdot n}-\sol\interleave_{Z,p}^2=O_p\Big(\interleave\varPi_\k\sol-\sol\interleave_{Z,p}^2+\tau_\k\interleave\cT\hsol_{\cdot n}-\cT\sol\interleave_W^2\Big).
\end{equation*}
Hence, the result follows by applying \eqref{lem:eq:2}.
\end{proof}

The following lemma is similar to Lemma A.2 of \cite{Breunig2012}. In the following, however, we provide the proof for the sake of completeness.
For all $\phi\in\cB$ recall the definition $g_j(\bold X_i,\phi)=(\1\{Y_i\leq \phi(Z_i)\}-q)f_j(W_i)$ for all $1\leq j\leq \m$ and $1\leq i\leq n$.
Let us introduce $\cX_{ii'}:=6\sqrt5/ (\sqrt\m n)\sum_{j=1}^\m \int g_j(\bold X_i,\solq)g_j(\bold X_{i'},\solq)dq$ and 
\begin{equation}\label{def:Q}
Q_{ni}:=
\left\{\begin{array}{lcl} 
\sum_{l=1}^{i-1}\cX_{li}, && \mbox{ for } i=2,\dots, n,\\
0,&&\mbox{ for } i=1 \mbox{ and } i>n.
\end{array}\right.
\end{equation}
Then clearly
\begin{multline*}
3\sqrt5/ (\sqrt\m n)\sum_{i\neq i'}\sum_{j=1}^\m \int g_j(\bold X_i,\solq)g_j(\bold X_{i'},\solq)dq\\\hfill
=6\sqrt5/ (\sqrt\m n)\sum_{i< i'}\sum_{j=1}^\m \int g_j(\bold X_i,\solq)g_j(\bold X_{i'},\solq)dq
=\sum_{i<i'}\cX_{ii'}=\sum_{i=1}^nQ_{ni}.
\end{multline*}
Let $\cB_{ni}:=\cB((Z_1,Y_1,W_1),\dots, (Z_i,Y_i,W_i))$, $1\leq i\leq n$, $n\geq 1$, be the $\sigma$-algebra generated by $(Z_1,Y_1,W_1),\dots, (Z_i,Y_i,W_i)$. Since $g_j(\bold X_i,\solq)$, $1\leq i\leq n$, are centered random variables it follows that $\{(\sum_{i'=1}^iQ_{ni'},\cB_{ni}), \,\,i\geq 1\}$ is a Martingale for each $n\geq 1$ and hence $\{(Q_{ni},\cB_{ni}), \,i\geq 1\}$ is a Martingale difference array for each $n\geq 1$. 
\begin{lem}\label{normal:lem:1}
 Let $Q_{ni}$ be defined as in \eqref{def:Q}. Let  Assumption \ref{ass:A1:q} and condition \eqref{cond:theo:norm} be satisfied. Then, we have $\sum_{i=1}^\infty Q_{ni}\stackrel{d}{\rightarrow} \cN(0,1)$. 
\end{lem}
\begin{proof}
For the proof we have to show that the Martingale difference array $\{(Q_{ni},\cB_{ni}),\,i\geq 1\}$, $n\geq 1$,  satisfies the conditions 
\begin{gather}
 \sum_{i= 1}^\infty\Ex|Q_{ni}|^2\leq 1 \quad \text{ for all }n\geq 1\label{normal:lem:1:eq1},\\
 \sum_{i=1}^\infty Q_{ni}^2=1+o_p(1)\label{normal:lem:1:eq2},\\
 \sup_{i\geq 1}|Q_{ni}|=o_p(1). \label{normal:lem:1:eq3}
\end{gather}
Then the result follows by \cite{Awad}.
Proof of \eqref{normal:lem:1:eq1}.
Since $\Ex[(\1\set{Y\leq \sol(Z,q)}-q)(\1\set{Y\leq \sol(Z,q')}-q')|W]=\min(q,q')-qq'$ we have
\begin{equation*}
\int \Big(\Ex\big[g_j(\bold X,\solq)g_{j'}(\bold X,\solqd)\big]\Big)^2d(q,q')
 =\int  (\min(q,q')-qq')^2d(q,q')\1_{\set{j=j'}}
 =\1_{\set{j=j'}}/90,
\end{equation*}
where we used that $\Ex[f_j(W)f_{j'}(W)]^2=\1_{\set{j=j'}}$  and 
\begin{align*}
 \int  (\min(q,q')-qq')^2d(q,q')
  &=\int\Big(\int_0^q (q'-qq')^2dq'+\int_q^1 (q-qq')^2dq'\Big)dq\\
 &=\frac{2}{3}\int q^3 (1-q)^2dq\\
 &=1/90.
\end{align*}
%
Observe that $\Ex[\cX_{1i}\cX_{1i'}]=0$ for $i\neq i'$ and thus, for $i=2,\dots,n$ we have
\begin{align*}
\Ex |Q_{ni}|^2&=\Ex|\cX_{1i}+\dots+\cX_{i-1,i}|^2\\
&=(i-1)\Ex|\cX_{12}|^2\\
&=\frac{(6\sqrt5)^2(i-1)}{n^2\m}\Ex\big|\sum_{j=1}^\m \int g_j(\bold X_1,\solq)g_j(\bold X_2,\solq)dq\big|^2\\
&=\frac{180\,(i-1)}{n^2\m}\sum_{j,j'=1}^\m\int \big(\Ex[ g_j(\bold X,\solq)g_{j'}(\bold X,\solqd)]\big)^2d(q,q')\\
&=\frac{2(i-1)}{n^2}.
\end{align*}
Thereby, we conclude
\begin{equation}\label{normal:lem:proof:eq1}
 \sum_{i= 1}^n\Ex|Q_{ni}|^2=\frac{2}{n^2}\sum_{i= 1}^{n-1}i=\frac{n(n-1)}{n^2}=1-\frac{1}{n}
\end{equation}
which proves \eqref{normal:lem:1:eq1}.

Proof of \eqref{normal:lem:1:eq2}. 
Using relation \eqref{normal:lem:proof:eq1} we observe
\begin{equation*}
 \Ex\big|\sum_{i=1}^nQ_{ni}^2-1\big|^2=\sum_{i=1}^n \Ex Q_{ni}^4+2\sum_{i< i'}\Ex Q_{ni}^2Q_{ni'}^2-1+o(1)=:I_n+II_n-1+o(1).
\end{equation*}
Consider $I_n$. 
Observe that
\begin{multline*}
 \Ex |Q_{ni}|^4=\Ex\big|\sum_{i'=1}^{i-1}\cX_{i'i}\big|^4\leq\int\Ex\Big|\frac{6\sqrt 5}{n\sqrt\m }\sum_{j=1}^\m g_j(\bold X_i,\solq)\sum_{i'=1}^{i-1}g_j(\bold X_{i'},\solq)\Big|^4dq\\\hfill
   \leq\frac{C}{n^4 m_n^2}\sup_{w\in\cW}\|f_\umn(w)\|^4\Big((i-1) \Ex\|f_\umn(W)\|^4+3(i-1)(i-2)(\Ex\|f_\umn(W)\|^2)^2\Big)
\end{multline*}
where we used that $\Ex[ g_j(\bold X,\solq)]=0$ for $0<q<1$. Since  $\sum_{i=1}^n3(i-1)(i-2)=n(n-1)(n-2)$ we conclude
\begin{equation*}
 I_n \leq C\Big(\frac{n(n-1)}{2n^4}\Ex\|f_\umn(W)\|^4+\frac{n(n-1)(n-2)}{n^4}(\Ex\|f_\umn(W)\|^2)^2\Big)=o(1)
\end{equation*}
since $(\Ex\|f_\umn(W)\|^2)^2\leq \Ex\|f_\umn(W)\|^4\leq Cm_n^2$. 
 We calculate  for $i<i'$
\begin{multline*}
 Q_{ni}^2Q_{ni'}^2
=\Big(\sum_{k=1}^{i-1}\cX_{ki}^2\Big)\Big(\sum_{k=1}^{i'-1}\cX_{ki'}^2\Big)
+\Big(\sum_{k=1}^{i-1}\cX_{ki}^2\Big)\Big(\sum_{k\neq k'}^{i'-1}\cX_{ki'}\cX_{k'i'}\Big)\\\hfill
+\Big(\sum_{k\neq k'}^{i-1}\cX_{ki}\cX_{k'i}\Big)\Big(\sum_{k=1}^{i'-1}\cX_{ki'}^2\Big)
+\Big(\sum_{k\neq k'}^{i-1}\cX_{ki}\cX_{k'i}\Big)\Big(\sum_{k\neq k'}^{i'-1}\cX_{ki'}\cX_{k'i'}\Big)\\
=:A_{ii'}+B_{ii'}+C_{ii'}+D_{ii'}.
\end{multline*}
Consider $A_{ii'}$. Exploiting relation \eqref{normal:lem:proof:eq1} and using
 $\sum_{i<i'}(i-1)
=\sum_{i'=1}^n(i'-1)(i'-2)/2=n(n-1)(n-2)/6$ and further  $\sum_{i<i'}(i-1)(i'-3)=\sum_{i'=1}^n(i'-3)(i'-2)(i'-1)/2=n(n-1)(n-2)(n-3)/8$  we obtain 
\begin{multline*}
 2\sum_{i<i'}\Ex A_{ii'}=4 \Ex \cX_{12}^2\cX_{23}^2\sum_{i<i'}(i-1) +2(\Ex \cX_{12}^2)^2\sum_{i<i'}(i-1)(i'-3)+o(1)\\
\leq C\frac{n(n-1)(n-2)}{n^4m_n^2}\sum_{j,l=1}^\m\int\Ex\big[ g_j^2(\bold X,\solq)g_l^2(\bold X,\sol_{q'})\big]d(q,q')\\+\frac{n(n-1)(n-2)(n-3)}{n^4}+o(1)
\end{multline*}
since $\int\Ex [g_j(\bold X,\solq)g_{j'}(\bold X,\solqd)]d(q,q')=\1\{j=j'\}/90$. 
Moreover, applying Cauchy Schwarz's inequality twice gives
\begin{equation*}
 \sum_{j,l=1}^\m\int\Ex\big[ g_j^2(\bold X,\solq)g_l^2(\bold X,\sol_{q'})\big]d(q,q')
\leq \sup_{w\in\cW}\|f_\umn(w)\|^4
\leq Cm_n^2.
\end{equation*}
Thereby, it holds $2\sum_{i<i'}\Ex A_{ii'}=1+o(1)$. 
Now consider $B_{ii'}$. Since $\{\basW_l\}_{l\geq 1}$ forms an orthonormal basis on the support of $W$ we obtain
\begin{multline*}
 \Ex \Big(\sum_{k=1}^{i-1}\cX_{ki}^2\Big)\Big(\sum_{k\neq k'}^{i'-1}\cX_{ki'}\cX_{k'i'}\Big)=2\sum_{k=1}^{i-1}\Ex \cX_{ki}^2\cX_{ki'}\c\cX_{ii'}\\
\leq \frac{C(i-1)}{n^4m_n ^2}\sum_{j,j'=1}^\m\int\Ex\Big|g_j(\bold X_1,\solq)g_{j'}(\bold X_1,\solq) g_j(\bold X_2,\solq)g_{j'}(\bold X_2,\solqd)\\\hfill
\times q(1-q)\sum_{l=1}^\m g_l^2(\bold X_1,\solq)\Big|d(q,q',q'')\\
\leq \frac{C(i-1) }{n^4\m}
\Big(\sum_{j,j'=1}^\m\int\Ex| g_j(\bold X,\solq)g_{j'}(\bold X,\solq)|^2d(q,q')\Big)
\leq \frac{C(i-1) m_n}{n^4}.
\end{multline*}
This, together with relation \eqref{normal:lem:proof:eq1}, yields $\sum_{i<i'}\Ex B_{ii'}=o(1)$.
Further, it is easily seen that $\sum_{i<i'}\Ex C_{ii'}=o(1)$. 
Consider $D_{ii'}$. Using twice the law of iterated expectation gives
\begin{multline*}
 \Ex D_{ii'}=\Ex\Big(\sum_{k\neq k'}^{i-1}\cX_{ki}\cX_{k'i}\Big)\Big(\sum_{k\neq k'}^{i'-1}\cX_{ki'}\cX_{k'i'}\Big)
  =4 \sum_{k<k'}^{i-1}\Ex \cX_{ki}\cX_{k'i}\cX_{ki'}\cX_{k'i'}\\
  =4 \sum_{k<k'}^{i-1}\Ex\big[ \cX_{ki}\cX_{k'i}\Ex[\cX_{ki'}\cX_{k'i'}|(Y_k,Z_k,W_k),(Y_{k'},Z_{k'},W_{k'}),(Y_{i},Z_{i},W_{i})]\big]\\
  \leq \frac{C}{n^2\m} \sum_{k<k'}^{i-1}\Ex\Big[ \Ex[\cX_{ki}\cX_{k'i}|(Y_k,Z_k,W_k),(Y_{k'},Z_{k'},W_{k'})]\\\hfill
  \times\sum_{j,j'=1}^\m\int \Ex[g_j(\bold X,\solq)g_{j'}(\bold X,\solqd)]g_j(\bold X_k,\solq)g_{j'}(\bold X_{k'},\solqd)d(q,q')\Big]\\\hfill
  \leq\frac{C}{n^4 m_n^2}\int \Ex\Big|\sum_{j,j'=1}^\m\Ex[g_j(\bold X,\solq)g_{j'}(\bold X,\solqd)]g_j(\bold X_1,\solq)g_{j'}(\bold X_2,\solqd)\Big|^2d(q,q')(i-1)(i-2)\\
  \leq \frac{C}{n^4\m}(i-1)(i-2).
\end{multline*} 
again using that $\Ex[g_j(\bold X,\solq)g_{j'}(\bold X,\solqd)]$ is only different from zero whenever $j=j'$. Consequently, we obtain
\begin{equation*}
 \sum_{i<i'}\Ex D_{ii'}\leq \frac{C}{n^4\m}\sum_{i<i'}(i-1)(i-2)
  = \frac{C\,n(n-1)(n-2)(n-3)}{\m n^4}=o(1)
\end{equation*} 
and hence $2\sum_{i< i'}\Ex Q_{ni}^2Q_{ni'}^2=1+o(1)$.

Proof of \eqref{normal:lem:1:eq3}. Note that $\PP\big(\sup_{i\geq 1}|Q_{ni}|>\varepsilon\big)\leq \sum_{i=1}^n\PP\big(Q_{ni}^2>\varepsilon^2\big)$ and, hence the assertion follows from the Markov inequality.
\end{proof}